\documentclass[11pt,a4paper]{article}
\usepackage[margin=2.5cm]{geometry}
\usepackage{amsmath,amssymb,amsthm,mathtools}
\usepackage{bm}
\usepackage[T1]{fontenc}
\usepackage[utf8]{inputenc}
\usepackage{lmodern}
\usepackage{microtype}
\usepackage[hidelinks]{hyperref}
\usepackage{authblk}
\usepackage{xcolor}

\newtheorem{theorem}{Theorem}[section]
\newtheorem{lemma}[theorem]{Lemma}
\newtheorem{proposition}[theorem]{Proposition}
\newtheorem{corollary}[theorem]{Corollary}
\theoremstyle{definition}
\newtheorem{definition}[theorem]{Definition}
\newtheorem{remark}[theorem]{Remark}
\newtheorem{example}[theorem]{Example}

\title{Variational Dissipative Mechanics on Lie Algebroids}
\author{Alexandre Anahory Simoes\thanks{School of Science and Technology, IE University, Madrid, Spain. (alexandre.anahory@ie.edu).} \, and Leonardo J. Colombo\thanks{Centro de Automática y Robótica (CSIC-UPM), Carretera de Campo Real, km 0, 200, 28500 Arganda del Rey, Spain. (leonardo.colombo@csic.es). The authors acknowledge financial support from Grant PID2022-137909-NB-C22 funded by the Spanish Ministry of Science and Innovation.}}
\date{\empty}

\begin{document}
\maketitle

\begin{center}
    \textit{This paper is dedicated to our friend Professor Juan Carlos Marrero\\ on the occasion of his 60th birthday.}
\end{center}

\begin{abstract}
We formulate a Herglotz-type variational principle on a Lie algebroid and derive the corresponding Euler--Lagrange--Herglotz equations for a Lagrangian  depending on an additional scalar variable $z$. 
This provides a geometric framework for dissipative systems on Lie algebroids
and recovers, as special cases, the classical Euler--Lagrange--Herglotz equations on tangent bundles, the Euler--Poincar\'e--Herglotz equations on a Lie algebra, and the 
Lagrange--Poincar\'e--Herglotz equations on Atiyah algebroids of principal bundles.
Starting from the local formulation, we then use Lie algebroid connections to obtain
a coordinate-free Euler--Lagrange--Poincar\'e--Herglotz and 
Hamilton--Pontryagin--Herglotz theory.
Finally, we establish energy balance laws and Noether--Herglotz-type results,
in which classical conserved quantities are replaced by dissipated invariants.
\end{abstract}

\section{Introduction}

Variational principles lie at the heart of classical and modern mechanics.
In the Lagrangian picture, the dynamics of a mechanical system with
configuration space $Q$ and Lagrangian $L \colon TQ \to \mathbb{R}$ is
obtained by requiring that the action functional
\[
S[q] = \int_0^T L\big(q(t),\dot q(t)\big)\,dt
\]
be stationary under variations of the curve $q(\cdot)$ with fixed endpoints.
The associated Euler--Lagrange equations provide a classical description of
conservative dynamics on $TQ$, and their reductions, under a symmetry group, yield the familiar
Euler--Poincar\'e and Lagrange--Poincar\'e equations \cite{holm1998euler}, \cite{cendra2001lagrangian}.

For systems with dissipation, the standard variational framework is
not directly applicable: dynamics under Rayleigh friction forces, thermodynamic
couplings, or interaction with a reservoir typically is not described by critical values of the classical action functional \cite{GayBalmazYoshimura2017PartI}.
One way to circumvent this difficulty, going back to Herglotz,
is to replace the action integral by an \emph{action variable}
$z(t)\in\mathbb{R}$ solving
\[
\dot z(t) = L\big(q(t),\dot q(t),z(t)\big),
\]
and to impose stationarity of the terminal value $z(T)$ under
variations of $q(\cdot)$.
The resulting Euler--Lagrange--Herglotz equations provide an alternative
variational description of a class of dissipative systems and
are closely related to contact geometry and contact Hamiltonian dynamics \cite{de2019contact}, \cite{de2020infinitesimal}, \cite{colombo2025liouville}, \cite{colombo2025homogeneous}.

On the other hand, Lie algebroids offer a unifying language for
geometric mechanics and reduction \cite{weinstein1996lagrangian}.
They encompass different pahse spaces such as the tangent bundle $TQ$, Lie algebras, action
algebroids and Atiyah algebroids of principal bundles. Thus, they form
a natural stage for symmetry-reduced
systems.
Lagrangian and Hamiltonian mechanics on Lie algebroids have been
developed in a series of works \cite{martinez2001lagrangian}, \cite{de2005lagrangian}, \cite{grabowska2006geometrical}, and the corresponding
Euler--Lagrange--Poincar\'e and Hamilton--Pontryagin formalisms have
been recently refined in the connection-based approach of
Li, Stern and Tang \cite{li2017lagrangian} and of Hu and Stern \cite{hu2024hamiltonian}.
In this framework, many reduced or constrained equations of motion
arise from a single geometric principle on an abstract algebroid.

The aim of this paper is to \emph{combine} these two viewpoints by
developing a \emph{Herglotz variational principle on a Lie algebroid}
$(E\to M,[\cdot,\cdot],\rho)$.
We consider a Herglotz-type Lagrangian $L \colon E \times \mathbb{R} \to \mathbb{R},\, (x,y,z)\mapsto L(x,y,z)$, and an additional scalar variable $z(t)$ evolving according to $\dot z(t) = L\big(x(t),y(t),z(t)\big)$, along an admissible curve $(x(t),y(t))$ in $E$.
Requiring stationarity of $z(T)$ under admissible variations
$E$-homotopic to $(x,y)$ leads to the system of
\emph{Euler--Lagrange--Herglotz equations on the Lie algebroid}, derived in \cite{anahory2024euler}. The variational principle and equations reduce to the classical Herglotz variational principle when
$E=TQ$.

Our first main result is the variational derivation of these
Euler--Lagrange--Herglotz equations on $(E,[\cdot,\cdot],\rho)$ in local coordinates. We then recall several familiar classes of equations that appear
as special cases such as as the classical Euler--Lagrange--Herglotz equations on $TQ$; the Hamel--Herglotz equations in quasi-velocities; the Euler--Poincar\'e--Herglotz equations on a Lie algebra and on an
        action Lie algebroid;
 and the Lagrange--Poincar\'e--Herglotz equations on the Atiyah algebroid
        of a principal bundle. In particular, we obtain a Herglotz-type, gauge-invariant dissipative
extension of Wong's equations on an Atiyah algebroid, together with a
thermomechanical interpretation for certain rigid body models with
entropy-like variables.

Our second main contribution is to analyze the \emph{energy balance} and
\emph{Noether-type} properties of these Herglotz systems on a Lie
algebroid.
We define a generalized Herglotz energy associated with $L$ and prove
a balance law showing that its evolution is governed by
$\partial L/\partial z$.
As a consequence, the energy is not conserved in general, but becomes
a \emph{dissipated quantity} in the sense of \cite{de2021symmetries} and,
after multiplication by a suitable integrating factor, it is
conserved along solutions.
An analogous result holds for Noether-type momenta associated with
infinitesimal symmetries of $L$.
We illustrate these ideas with examples including mechanical systems
with Rayleigh dissipation, dissipative Wong equations, and rigid bodies
coupled to entropy-like variables.

The third main ingredient of the paper is a \emph{coordinate-free}
reformulation of the Herglotz equations on a Lie algebroid, based on
$TM$-connections and the Euler--Lagrange--Poincar\'e viewpoint of
Li--Stern--Tang and Hu--Stern.
Given a $TM$-connection $\nabla$ on $E$ and the induced $E$-connection
$\nabla^*$ on $E^*$, we show that the local Euler--Lagrange--Herglotz
equations are equivalent to the intrinsic relation
\[
\overline{\nabla}^*_{a(t)}p(t)
-
\rho^*\!\big(dL_{\mathrm{hor}}(a(t),z(t))\big)
=\frac{\partial L}{\partial z}\big(a(t),z(t)\big)\,p(t),
\qquad
p = dL_{\mathrm{ver}}(a,z),
\]
for the momentum $p\in E^*$ along an admissible curve $a(t)\in E$.

Finally, we extend the Hamilton--Pontryagin principle on Lie algebroids to
a \emph{Hamilton--Pontryagin--Herglotz} principle, where the
Herglotz variable $z$ satisfies a differential equation including both the Lagrangian and a Lagrange multiplier, dual variable $p$, responsible for ensuring the additional constraint that a trajectory must be an $E$-path. The resulting implicit Euler--Lagrange--Poincar\'e--Herglotz equations
provide a convenient starting point for discretization and structure-preserving
numerical schemes on Lie groupoids. Several nontrivial examples are presented to illustrate the scope of the
theory. 

The paper is organized as follows.
Section~\ref{sec:preliminaries} recalls basic notions on Lie algebroids
and admissible curves.
In Section~\ref{sec:Herglotz_LA}, we introduce the Herglotz variational
principle on a Lie algebroid and derive the
Euler--Lagrange--Herglotz equations, including several classical
examples and reduction-type equations.
Section~\ref{section5} is devoted to the generalized energy
balance law and a Noether--Herglotz theorem, together with examples of
dissipated quantities.
Finally, in Sections~\ref{sec:coordinate_free} and \ref{sec:HPH}, we develop the
Euler--Lagrange--Poincar\'e--Herglotz equations in a connection-based,
coordinate-free form and formulate the Hamilton--Pontryagin--Herglotz
principle on a Lie algebroid, respectively.

\section{Preliminaries on Lie Algebroids}\label{sec:preliminaries}


Let $\tau:E \to M$ be a vector bundle of rank $r$ over a manifold $M$ of
dimension $n$. Denote by $\Gamma(E)$ the set of sections of this vector bundle, by $\mathfrak{X}(M)$ the set of vector fields on $M$, and by $C^{\infty}(M)$ the set of smooth functions on $M$. If $\gamma$ is a curve on $M$, we will denote by $\Gamma_{\gamma}(E)$ the  set of sections of $\Gamma(E)$ evaluated along the curve $\gamma$.

A \emph{Lie algebroid} structure on $E$ consists of a Lie bracket $[\cdot,\cdot]\colon \Gamma(E)\times\Gamma(E)\to\Gamma(E)$ and a vector bundle map over the identity $\rho\colon E\to TM$, or $\rho\colon \Gamma(E)\to \mathfrak{X}(M)$, called the anchor, satisfying the following Leibniz rule
\[
[\sigma_1, f\sigma_2] = f[\sigma_1,\sigma_2] + (\rho(\sigma_1)f)\,\sigma_2,
\]
for all $\sigma_1,\sigma_2\in\Gamma(E)$ and $f\in C^\infty(M)$.

Locally, let $(x^i)$, $i=1,\dots,n$ be coordinates on $M$, and 
$\{e_\alpha\}$, $\alpha=1,\dots,r$ a local basis of sections of $E$. Then the anchor and the bracket satisfy
\[
\rho(e_\alpha) = \rho^i_{\ \alpha}(x)\,\frac{\partial}{\partial x^i}, 
\qquad
[e_\alpha,e_\beta] = C^\gamma_{\ \alpha\beta}(x)\, e_\gamma,
\]
for some smooth functions $\rho^i_{\ \alpha}$ and
$C^\gamma_{\ \alpha\beta}$ caled the anchor components and the structure functions, respectively. The Jacobi identity and the
Leibniz rule impose useful identities on these functions.


The local basis of sections $\{e_{\alpha}\}$ induces the fibred coordinates $(x^{i},y^{\alpha})$ on $E$ adapted to $\tau$. A curve $a\colon [0,T] \to E$, written in local coordinates as
\[
a(t) = \big(x^i(t), y^\alpha(t)\big), \qquad
a(t) = y^\alpha(t) e_\alpha\big|_{x(t)},
\]
is called \emph{admissible} or an \textit{$E$-path} if its base curve $x(t)=\tau(a(t))$ satisfies
\begin{equation}\label{eq:admissible}
\dot x^i(t) = \rho^i_{\ \alpha}\big(x(t)\big)\, y^\alpha(t).
\end{equation}
The space of admissible curves generalizes the space of velocities
$(q(t),\dot q(t))$ in $TQ$.


Following the standard variational calculus on Lie algebroids \cite{martinez2008variational}, we describe
variations of admissible curves via \emph{$E$-homotopies}.

\begin{definition}
An \emph{$E$-homotopy} is a smooth map
\[
a \colon (-\varepsilon,\varepsilon)\times [0,T] \to E, \quad 
(s,t)\mapsto a_s(t),
\]
such that for each $s$, the curve $t\mapsto a_s(t)$ is admissible, i.e.
$\dot x_s^i(t)= \rho^i_{\ \alpha}(x_s(t))\,y_s^\alpha(t)$, where
$a_s(t) = (x_s(t),y_s(t))$.
We say that $a_s$ is a variation of the reference curve $a_0$.
\end{definition}

Define the variation field $\frac{\partial a}{\partial s}(0,t)\in T_{a_{0}(t)}E$ to be given in local coordinates by 
\[
\delta x^i(t) := \left.\frac{\partial x_s^i(t)}{\partial s}\right|_{s=0},\qquad
\delta y^\alpha(t) := \left.\frac{\partial y_s^\alpha(t)}{\partial s}\right|_{s=0}.
\]

A convenient way to parameterize variations of the reference curve $a_{0}(t)$ is via a time-dependent
section $\xi(t) \in \Gamma(E)$ along $x_{0}(t)$, written locally as
$\xi(t) = \xi^\alpha(t)\,e_\alpha|_{x_{0}(t)}$. Following \cite{martinez2008variational}, for an admissible $E$-homotopy the variation field satisfies
\[
\delta x^i = \rho^i_{\ \alpha}(x)\,\xi^\alpha.
\]
and
\[
\delta y^\gamma = \dot{\xi}^\gamma + C^\gamma_{\ \alpha\beta}(x)\,y^\alpha\,\xi^\beta.
\]



In the variational principle presented in the next section, we will consider variations with fixed endpoints of the base curve $x(t)$. Locally, these are the variations of the form above satisfying $\xi^\alpha(0)=\xi^\alpha(T)=0$. 

\begin{remark}
    In Section~3 of Martínez \cite{martinez2008variational}
the definition of admissible variations is given in an intrinsic, global form.
An $E$-homotopy is described as a Lie algebroid morphism
\[
\Phi = a\, dt + b\, ds : T(I\times J) \longrightarrow E,
\]
subject to the usual admissibility conditions and boundary constraints.
Given an admissible curve $a(t)$, the infinitesimal variation associated with a
time-dependent section $\sigma(t)\in \Gamma(E)$ along the base curve
$\gamma(t) = \tau(a(t))$
is defined through the canonical morphism
\[
\Xi_a : \Gamma_\gamma(E) \longrightarrow \Gamma_a(TE),
\qquad
\Xi_a(\sigma)(t)
= \rho_1\!\big(\chi_E(\sigma(t), a(t), \dot{\sigma}(t))\big),
\]
where $\chi_E : \mathcal{T}^{E}E\to \mathcal{T}^{E}E$ denotes the canonical involution in the prolongation algebroid $\mathcal{T}^{E}E$
and $\rho_1$ is the anchor of the vector bundle $\mathcal{T}^{E}E\to E$.
In local coordinates $(x^i, y^\alpha)$ this yields
\begin{equation}
\Xi_a(\sigma)(t)
= \rho^i_{\alpha}(\gamma(t))\,\sigma^\alpha(t)\,
\frac{\partial}{\partial x^i}
+ \Big(
\dot{\sigma}^\alpha(t)
+ C^{\alpha}_{\beta\gamma}(\gamma(t))\,a^\beta(t)\,\sigma^\gamma(t)
\Big)
\frac{\partial}{\partial y^\alpha}.
\label{eq:Xi_a_local}
\end{equation}
This expression provides a global, coordinate-free characterization
of admissible variation fields on a Lie algebroid.

The definition used in this paper corresponds to the
local expression of the operator $\Xi_a$ above.
There, the variation is introduced directly through
a time-dependent section $\xi(t)$ along $x(t)$, leading to
\[
\delta x^i = \rho^i_\alpha(x)\,\xi^\alpha,
\qquad
\delta y^\gamma = \dot{\xi}^\gamma
+ C^\gamma_{\alpha\beta}(x)\,y^\alpha\,\xi^\beta.
\]
Comparing with~\eqref{eq:Xi_a_local}, we see that
\[
\Xi_a(\xi) = \delta x^{i}\frac{\partial}{\partial x^i} + \delta y^{\alpha}\frac{\partial}{\partial y^{\alpha}},
\]
so both definitions are equivalent in content.
The formulation in Martínez's paper \cite{martinez2008variational} is fully geometric
and independent of coordinates, while the one used here
is its local version, more convenient for explicit computations
in the Herglotz variational setting.\hfill$\diamond$
\end{remark}

\section{Herglotz Variational Principle on a Lie Algebroid}\label{sec:Herglotz_LA}


Troughout the text, let $(x^{i},y^{\alpha})$ be fibred coordinates on $E$ associated to the projection $\tau$ and to a choice of a local basis of sections. Throughout the text the algebroid structure of $\tau:E\to M$ will be naturally extended to an algebroid over $E\times \mathbb{R}\to Q$ through the projection onto the first component. The Lie algebroid projection, the anchor map and the bracket of sections and all local functions will be written using the same notation.

Let $L \colon E \times \mathbb{R} \to \mathbb{R}$ be a smooth function, which we call a \emph{Herglotz-type Lagrangian}.
We introduce an additional scalar variable $z(t)\in\mathbb{R}$ satisfying the Herglotz differential equation
\begin{equation}\label{eq:herglotz_z}
\dot z(t) = L\big(x^{i}(t),y^{\alpha}(t),z(t)\big).
\end{equation}

An admissible trajactory solving the Herglotz equation is a triple $t\mapsto \big(x^{i}(t),y^{\alpha}(t),z(t)\big)\in E\times\mathbb{R}$,
such that:
\begin{equation}\label{eq:herglotz_admissible}
\dot x^i(t) = \rho^i_{\ \alpha}(x^{i}(t))\,y^\alpha(t), \qquad
\dot z(t) = L(x^{i}(t),y^{\alpha}(t),z(t)).
\end{equation}

\begin{definition}
We say that a triple $(x^{i}(t),y^{\alpha}(t),z(t))$ is a \emph{Herglotz extremal} on
the Lie algebroid if, for any $E$-homotopy $(x^{i}_s(t),y^{\alpha}_s(t),z_s(t))$
such that:
\begin{enumerate}
   \item for each $s$, $(x^{i}_s,y^{\alpha}_s,z_s)$ satisfies \eqref{eq:herglotz_admissible};
  \item $x^{i}_s(0)=x^{i}(0)$, $x^{i}_s(T)=x^{i}(T)$ for all $s$;
  \item $z_s(0)=z(0)$ for all $s$;
\end{enumerate}
we have the stationarity condition
\begin{equation}\label{eq:stationary}
\left.\frac{d}{ds}\right|_{s=0} z_s(T) = 0.
\end{equation}
\end{definition}

In other words, $z(T)$ plays the role of an action functional, and we require
it to be stationary under admissible $E$-homotopies.


\begin{theorem}[Euler--Lagrange--Herglotz equations on a Lie algebroid]
\label{thm:ELH_algebroid}
Let $(E\to M,[\cdot,\cdot],\rho)$ be a Lie algebroid whose anchor components are given by
$\rho^i_{\ \alpha}(x)$ and strcuture functions are $C^\gamma_{\ \alpha\beta}(x)$ with respect to local coordinates denoted by $(x^{i},y^{\alpha},z)$.
Let $L\colon E\times\mathbb{R}\to\mathbb{R}$ be a Herglotz-type Lagrangian and consider a curve $(x^{i}(t),y^{\alpha}(t),z(t))$ satisfying
\begin{equation}\label{eq:ELH_constraints}
\dot x^i(t) = \rho^i_{\ \alpha}(x^{i}(t))\,y^\alpha(t), 
\qquad
\dot z(t) = L(x^{i}(t),y^{\alpha}(t),z(t)).
\end{equation}
Then $(x^{i}(t),y^{\alpha}(t),z(t))$ is a Herglotz extremal, i.e. satisfies
\eqref{eq:stationary} for all admissible variations with $x^{i}(0),x^{i}(T),z(0)$ fixed, 
if and only if it satisfies the Euler--Lagrange--Herglotz equations:
\begin{equation}\label{eq:ELH_equations}
\frac{d}{dt}\left(\frac{\partial L}{\partial y^\alpha}\right)
+ C^\gamma_{\ \alpha\beta}(x)\,y^\beta \frac{\partial L}{\partial y^\gamma}
- \rho^i_{\ \alpha}(x)\,\frac{\partial L}{\partial x^i}
- \frac{\partial L}{\partial z}\,\frac{\partial L}{\partial y^\alpha}
= 0.
\end{equation}
\end{theorem}

\begin{proof}
Let $(x^{i}_s(t),y^{\alpha}_s(t),z_s(t))$ be an admissible variation, i.e.
for each $s$, $\dot x_s^i(t) = \rho^i_{\ \alpha}(x_s(t))\,y_s^\alpha(t)$ and $
\dot z_s(t) = L(x^{i}_s(t),y^{\alpha}_s(t),z_s(t))$, with $x^{i}_s(0)=x^{i}(0)$, $x^{i}_s(T)=x^{i}(T)$, $z_s(0)=z(0)$.
Define $\delta x^i = \partial_s x_s^i|_{s=0}$, 
$\delta y^\alpha=\partial_s y_s^\alpha|_{s=0}$, 
$\delta z = \partial_s z_s|_{s=0}$. Differentiating the Herglotz equation with respect to $s$:
\[
\frac{d}{dt}(\delta z)
= \frac{\partial L}{\partial x^i}\,\delta x^i
+ \frac{\partial L}{\partial y^\alpha}\,\delta y^\alpha
+ \frac{\partial L}{\partial z}\,\delta z.
\]
This is a linear ODE in $\delta z$. Rearranging:
\[
\frac{d}{dt}(\delta z) - \frac{\partial L}{\partial z}\,\delta z
= \frac{\partial L}{\partial x^i}\,\delta x^i
+ \frac{\partial L}{\partial y^\alpha}\,\delta y^\alpha.
\]
Let
\[
\lambda(t) := \exp\!\left(-\int_0^t \frac{\partial L}{\partial z}(\omega)\,d\omega\right).
\]
Then
\[
\frac{d}{dt} \big(\lambda \delta z\big)
= \lambda \left[\frac{d}{dt}(\delta z) - \frac{\partial L}{\partial z}\,\delta z\right]
= \lambda \left(
\frac{\partial L}{\partial x^i}\,\delta x^i
+ \frac{\partial L}{\partial y^\alpha}\,\delta y^\alpha
\right).
\]
Integrating both sides of the previous equation from $0$ to $T$ and using $\delta z(0)=0$ (since $z(0)$ is fixed),
we get
\[
\lambda(T)\,\delta z(T)
= \int_0^T \lambda(t)
\left(
\frac{\partial L}{\partial x^i}\,\delta x^i
+ \frac{\partial L}{\partial y^\alpha}\,\delta y^\alpha
\right) \,dt.
\]
Since $\lambda(T)\neq 0$, the stationarity condition 
$\delta z(T)=0$ is equivalent to
\begin{equation}\label{eq:stationary_integral}
\int_0^T \lambda(t)
\left(
\frac{\partial L}{\partial x^i}\,\delta x^i
+ \frac{\partial L}{\partial y^\alpha}\,\delta y^\alpha
\right) \,dt = 0.
\end{equation}

Now, we express $(\delta x^{i},\delta y^{\alpha})$ in terms of a time-dependent
section $\xi(t) = \xi^\alpha(t)\,e_\alpha|_{x(t)}$ along $x(t)$, via:
\[
\delta x^i(t) = \rho^i_{\ \alpha}(x^{i}(t))\,\xi^\alpha(t),
\qquad
\delta y^\gamma (t)= \dot \xi^\gamma (t) + C^\gamma_{\ \alpha\beta}(x^{i}(t))\,y^\alpha (t) \,\xi^\beta (t).
\]
Substituting into \eqref{eq:stationary_integral}, we obtain
\[
\int_0^T \lambda(t)
\left[
\frac{\partial L}{\partial x^i}\,\rho^i_{\ \alpha}\,\xi^\alpha
+ \frac{\partial L}{\partial y^\gamma}
\left(
\dot \xi^\gamma + C^\gamma_{\ \alpha\beta}y^\alpha\,\xi^\beta
\right)
\right] \,dt = 0.
\]
where we drop the functions inputs to simplify the reading. Rewriting:
\[
\int_0^T \left[
\lambda \frac{\partial L}{\partial x^i}\,\rho^i_{\ \alpha}\,\xi^\alpha
+ \lambda \frac{\partial L}{\partial y^\gamma} C^\gamma_{\ \alpha\beta}y^\alpha\,\xi^\beta
+ \lambda \frac{\partial L}{\partial y^\gamma} \dot \xi^\gamma
\right] dt = 0.
\]

We integrate by parts the term with $\dot \xi^\gamma$:
\[
\int_0^T \lambda \frac{\partial L}{\partial y^\gamma} \dot \xi^\gamma\, dt
= \left[ \lambda \frac{\partial L}{\partial y^\gamma} \xi^\gamma \right]_0^T
- \int_0^T \frac{d}{dt}\Big(\lambda \frac{\partial L}{\partial y^\gamma}\Big)
\,\xi^\gamma \,dt.
\]
The boundary term vanishes because $\xi^{\gamma}(0)=\xi^{\gamma}(T)=0$ (fixed endpoints).
Hence the stationarity condition becomes
\[
\int_0^T \left\{
\lambda \frac{\partial L}{\partial x^i}\,\rho^i_{\ \alpha}
- \lambda \frac{\partial L}{\partial y^\gamma} C^\gamma_{\ \alpha\beta}y^\beta
- \frac{d}{dt}\Big(\lambda \frac{\partial L}{\partial y^\alpha}\Big)
\right\} \xi^\alpha \, dt = 0.
\]
where the middle term changes sign due to the skew-symmetry of the structure functions, i.e., $C_{\alpha \beta}^{\gamma}=-C_{\beta \alpha}^{\gamma}$. Since $\xi^\alpha(t)$ are arbitrary functions vanishing at $t=0,T$,
we deduce the pointwise condition
\[
- \frac{d}{dt}\Big(\lambda \frac{\partial L}{\partial y^\alpha}\Big)
+ \lambda \frac{\partial L}{\partial x^i}\,\rho^i_{\ \alpha}
- \lambda \frac{\partial L}{\partial y^\gamma} C^\gamma_{\ \alpha\beta}y^\beta
= 0.
\]
Expanding the derivative:
\[
\frac{d}{dt}\Big(\lambda \frac{\partial L}{\partial y^\alpha}\Big)
= \dot \lambda \frac{\partial L}{\partial y^\alpha}
+ \lambda \frac{d}{dt}\left( \frac{\partial L}{\partial y^\alpha}\right).
\]
and using $\dot\lambda = -(\partial L/\partial z)\lambda$ by definition of $\lambda$,
we obtain
\[
\dot \lambda \frac{\partial L}{\partial y^\alpha}
= -\lambda \frac{\partial L}{\partial z}\frac{\partial L}{\partial y^\alpha}.
\]
Hence,
\[
- \left[
-\lambda \frac{\partial L}{\partial z}\frac{\partial L}{\partial y^\alpha}
+ \lambda \frac{d}{dt}\left( \frac{\partial L}{\partial y^\alpha}\right)
\right]
+ \lambda \frac{\partial L}{\partial x^i}\,\rho^i_{\ \alpha}
- \lambda \frac{\partial L}{\partial y^\gamma} C^\gamma_{\ \alpha\beta}y^\beta
= 0.
\]
Simplifying and dividing by $\lambda$, which is a nowhere vanishing function, we conclude
\[
\frac{d}{dt}\left(\frac{\partial L}{\partial y^\alpha}\right)
+ C^\gamma_{\ \alpha\beta}(x)\,y^\beta \frac{\partial L}{\partial y^\gamma}
- \rho^i_{\ \alpha}(x)\,\frac{\partial L}{\partial x^i}
- \frac{\partial L}{\partial z}\,\frac{\partial L}{\partial y^\alpha}
= 0.
\]
These are the desired Euler--Lagrange--Herglotz equations \eqref{eq:ELH_equations}.

The converse (that any solution of \eqref{eq:ELH_constraints}--\eqref{eq:ELH_equations}
is stationary) follows by reversing the above computations, showing that then
the integral \eqref{eq:stationary_integral} vanishes for all admissible
variations.
\end{proof}

\begin{remark}
It is worth emphasizing that the Euler--Lagrange--Herglotz equations 
\eqref{eq:ELH_equations} admit a fully geometric interpretation
in terms of the contact and Jacobi structures developed in 
\cite{anahory2024euler,anahory2025contact}.  
In this sense, Theorem~\ref{thm:ELH_algebroid} provides the local expression of the same intrinsic Herglotz dynamics obtained in those works via 
prolongations and contact Lie algebroids.

More precisely, in \cite{anahory2024euler} the authors construct a canonical Jacobi structure on $E^{*}\times\mathbb{R}$ and derive the contact Lagrangian equations pulling back the Hamiltonian equations to $E\times \mathbb{R}$.

Likewise, the contact Lie algebroid formalism in \cite{anahory2025contact} derives the intrinsic Herglotz equations using the prolongation 
$\mathcal{T}^{E}(E\times\mathbb{R})$ and the algebroid differential 
$\mathrm{d} = d^{\mathcal{T}^{E}(E\times\mathbb{R})}$.  
In particular, the solutions of the Herglotz equations are admissible trajectories of a SODE section $\Gamma_{L}:E\times \mathbb{R} \to \mathcal{T}^{E}(E\times\mathbb{R})$ determined by the equations
\[
\iota_{\Gamma_{L}}\,\mathrm{d}(\eta_{L})
 = \mathrm{d}(E_{L}) - R_{L}(E_{L})\,\eta_{L},
\qquad
\iota_{\Gamma_{L}}\eta_{L}=-E_{L},
\]
where $\eta_{L}:E\times \mathbb{R} \to (\mathcal{T}^{E}(E\times\mathbb{R}))^{*}$ is the contact section associated with $L$, the function $E_{L}$ is the corresponding energy, the section $R_{L}:E\times \mathbb{R} \to \mathcal{T}^{E}(E\times\mathbb{R})$ is the Reeb section uniquely determined by the conditions 
$\iota_{R_{L}}\,\mathrm{d} \eta_{L}=0$ and $\iota_{R_{L}}\eta_{L}=1$, and satisfying $R_{L}(E_{L})=-\frac{\partial L}{\partial z}$.


Therefore, Theorem~\ref{thm:ELH_algebroid} complements the geometric constructions of 
\cite{anahory2024euler,anahory2025contact} by providing a coordinate variational formulation of the Herglotz dynamics.\hfill$\diamond$
\end{remark}


\begin{example}[Euler--Lagrange--Herglotz equations]

Let $E=TQ$ be the tangent bundle of a configuration manifold $Q$.
In this case, we choose local coordinates $(q^i)$ and induced coordinates
$(q^i,\dot q^i)$ on $TQ$, i.e., $x^i = q^i, y^i = \dot q^i$, and the Lie algebroid structure is trivial. Then the Lagrangian is $L=L(q,\dot q,z)$, and the Euler--Lagrange--Herglotz equations \eqref{eq:ELH_equations} become
\[
\frac{d}{dt}\left(\frac{\partial L}{\partial \dot q^i}\right)
- \frac{\partial L}{\partial q^i}
- \frac{\partial L}{\partial z}\,\frac{\partial L}{\partial \dot q^i}
= 0, \quad \dot z = L(q,\dot q,z),
\]
which are the classical Euler--Lagrange--Herglotz equations.
\end{example}

\begin{example}[Hamel-Herglotz equations on Lie algebroids]
The Euler--Lagrange--Herglotz equations on a Lie algebroid 
$(E,\rho,[\cdot,\cdot])$, constitute the general formulation of Herglotz mechanics on a Lie algebroid. The classical Hamel--Herglotz equations arise as the special case 
$E=TQ$, where one chooses a (possibly nonholonomic) local frame 
$\{e_\alpha\}$ on $TQ$ with quasi-velocities 
$y^\alpha = a_i^\alpha(q)\,\dot q^i$.
In this case the anchor and the structure functions satisfy
\[
\rho_\alpha^{\ i}=a_\alpha^{\ i}, \qquad 
C_{\alpha\beta}^{\ \ \gamma}=\Gamma_{\alpha\beta}^{\ \ \gamma},
\]
where $\Gamma_{\alpha\beta}^{\ \ \gamma}$ are the commutation functions 
of the frame.
Substituting these expressions into 
\eqref{eq:ELH_equations} yields exactly the
Hamel--Herglotz equations
\begin{equation}
\frac{d}{dt}\!\left(\frac{\partial L}{\partial y^\alpha}\right)
- a_\alpha^i\,\frac{\partial L}{\partial q^i}
+ \Gamma_{\alpha\beta}^{\ \ \gamma} y^\beta 
\frac{\partial L}{\partial y^\gamma}
- \frac{\partial L}{\partial z}\,\frac{\partial L}{\partial y^\alpha}=0.
\end{equation}

\end{example}

\begin{example}[Euler--Poincar\'e--Herglotz equation on a Lie algebra]

Let $E=\mathfrak{g}$ be a Lie algebra, regarded as a Lie algebroid over a point.
Thus $M=\{\ast\}$, the anchor vanishes, and an admissible curve in $E$ is simply
a curve $t \mapsto \xi(t) \in \mathfrak{g}$. Choose a basis $\{e_\alpha\}$ of $\mathfrak{g}$ with brackets
$[e_\alpha,e_\beta]=C_{\alpha\beta}^{\ \ \gamma}e_\gamma$.

A Herglotz Lagrangian on $\mathfrak{g}$ is a function
$\ell(\xi,z)$ with $(\xi,z)\in \mathfrak{g}\times\mathbb{R}$. Writing $\xi=\xi^\alpha e_\alpha$, the
Euler--Lagrange--Herglotz equations reduce to
\[
\frac{d}{dt}\!\left(\frac{\partial \ell}{\partial \xi^\alpha}\right)
+ C_{\alpha\beta}^{\ \ \gamma}\,\xi^\beta 
   \frac{\partial \ell}{\partial \xi^\gamma}
= 
\frac{\partial \ell}{\partial z}\,
   \frac{\partial \ell}{\partial \xi^\alpha},
\qquad
\dot z = \ell(\xi,z).
\]

Identifying the covector 
\[
\frac{\partial \ell}{\partial \xi}
:= \frac{\partial \ell}{\partial \xi^\alpha}\,e^{\alpha}
\in \mathfrak{g}^*,
\]
the equation takes the intrinsic Euler--Poincar\'e--Herglotz form \cite{anahory2024reduction}
\[
\frac{d}{dt}\left(\frac{\partial \ell}{\partial \xi}\right)
+ \operatorname{ad}^*_{\xi}\!\left(\frac{\partial \ell}{\partial \xi}\right)
= 
\frac{\partial \ell}{\partial z}\,
\frac{\partial \ell}{\partial \xi}\qquad
\dot z = \ell(q,\xi,z).
\]

\end{example}

\begin{example}[Euler--Poincar\'e--Herglotz equations on an action Lie algebroid]\label{ex:EPH_action_intrinsic}

Let $G$ be a Lie group with Lie algebra $\mathfrak{g}$ acting on a manifold
$Q$ from the left. The associated action Lie algebroid is $E = Q \times \mathfrak{g} \;\to\; Q$, with anchor and bracket given by $\rho(q,\xi) = \xi_Q(q)$, $[(q,\xi),(q,\eta)] = (q,[\xi,\eta])$, where $\xi_Q$ denotes the infinitesimal generator of $\xi\in\mathfrak{g}$.

Choose a basis $\{e_\alpha\}$ of $\mathfrak{g}$, with
$[e_\alpha,e_\beta]=C_{\alpha\beta}^{\ \ \gamma}e_\gamma$, and write
$\xi = \xi^\alpha e_\alpha$. The anchor in local coordinates reads
\[
\rho(q,e_\alpha) = \rho_\alpha^i(q)\,\frac{\partial}{\partial q^i}
= (e_\alpha)_Q(q),
\]
so that the admissibility condition for a curve $(q(t),\xi(t))$ becomes $\dot q^i(t) = \rho_\alpha^{\,i}\big(q(t)\big)\,\xi^\alpha(t)$.

Taking the Herglotz Lagrangian on the action Lie algebroid $\ell(q,\xi,z)$ with $(q,\xi,z)\in Q\times\mathfrak{g}\times\mathbb{R}$, the Euler--Lagrange--Herglotz equations
\eqref{eq:ELH_equations} become
\begin{equation}\label{eq:EPH_action_coords}
\frac{d}{dt}\!\left(\frac{\partial \ell}{\partial \xi^\alpha}\right)
+ C_{\alpha\beta}^{\ \ \gamma}\,\xi^\beta 
   \frac{\partial \ell}{\partial \xi^\gamma}
- \rho_\alpha^{\,i}(q)\,\frac{\partial \ell}{\partial q^i}
=
\frac{\partial \ell}{\partial z}\,
   \frac{\partial \ell}{\partial \xi^\alpha},
\qquad
\dot z = \ell(q,\xi,z),
\end{equation}
together with the kinematic relation $\dot q^i = \rho_\alpha^{\,i}(q)\,\xi^\alpha$.


Let $J\colon T^*Q\to\mathfrak{g}^*$ be the standard momentum map for the
cotangent-lifted action, characterized by
$\big\langle J(q,p_q),\eta\big\rangle 
= \big\langle p_q,\eta_Q(q)\big\rangle$ with $\eta\in\mathfrak{g}$. Then, \eqref{eq:EPH_action_coords} can be written intrinsically as the
Euler--Poincar\'e--Herglotz equations on the action Lie algebroid
\begin{equation}\label{eq:EPH_action_intrinsic}
\frac{d}{dt}\left(\frac{\partial \ell}{\partial \xi}\right)
+ \operatorname{ad}^*_{\xi}\!\left(\frac{\partial \ell}{\partial \xi}\right)
= J\!\left(\frac{\partial \ell}{\partial q}\right)
+ \frac{\partial \ell}{\partial z}\,\frac{\partial \ell}{\partial \xi},
\end{equation}
with
\[
\dot q = \xi_Q(q), \qquad \dot z = \ell(q,\xi,z).
\]
\end{example}

\begin{example}[Lagrange--Poincaré--Herglotz equations on the Atiyah algebroid]
\label{example:LPH-Atiyah}

Let $\pi\colon Q\to M:=Q/G$ be a principal $G$-bundle, with $\dim M=n$ and 
$\dim G=d$.  
The quotient $\widehat{TQ}=TQ/G\to M$ is the Atiyah algebroid.  
Choose a local trivialization of $Q$ and let 
$\{e_i,\widehat e_A\}$ be the induced local $G$-invariant frame of $TQ/G$ 
(see \cite{anahory2024euler}, \cite{simoes2024symmetry}, \cite{anahory2025contact} for details).  
Then the anchor and bracket are
\[
\rho(e_i)=\frac{\partial}{\partial q^i},
\qquad 
\rho(\widehat e_A)=0,
\]
\[
[e_i,e_j] = -\mathcal{B}_{ij}^{A}\,\widehat e_A,\qquad
[e_i,\widehat e_A] = c_{AB}^{\ \ C}\,\mathcal{A}_{i}^{B}\,\widehat e_C,\qquad
[\widehat e_A,\widehat e_B]=c_{AB}^{\ \ C}\,\widehat e_C,
\]
where $\mathcal{A}_i^A$ and $\mathcal{B}_{ij}^A$ are the local coefficients of the principal 
connection and its curvature, and $c_{AB}^{\ \ C}$ are the structure constants of $\mathfrak g$.

Let $L=L(q^i,\dot q^i,v^A,z)\colon \widehat{TQ}\times\mathbb{R}\to\mathbb{R}$ be the reduced Herglotz Lagrangian induced by a $G$-invariant Lagrangian on 
$TQ\times\mathbb{R}$.  
Applying Theorem~\ref{thm:ELH_algebroid} to this Lie algebroid yields the 
Lagrange--Poincaré--Herglotz equations:
\begin{align}
\frac{\partial L}{\partial q^{j}}
 - \frac{d}{dt}\!\left(\frac{\partial L}{\partial \dot{q}^{j}}\right)
&=
\frac{\partial L}{\partial v^{A}}
\left(\mathcal{B}_{ij}^{A}\,\dot{q}^{i}
+ c_{DB}^{\ \ A}\,\mathcal{A}_{j}^{B}\,v^{D}\right)
 - \frac{\partial L}{\partial z}\,\frac{\partial L}{\partial \dot{q}^{j}},
\qquad j=1,\dots,n,
\label{eq:LPH-horizontal-Atiyah}
\\[0.4em]
\frac{d}{dt}\!\left(\frac{\partial L}{\partial v^{B}}\right)
&=
\frac{\partial L}{\partial v^{A}}
\left(c_{DB}^{\ \ A}\,v^{D}
 - c_{DB}^{\ \ A}\,\mathcal{A}_{i}^{D}\,\dot{q}^{i}\right)
 + \frac{\partial L}{\partial z}\,\frac{\partial L}{\partial v^{B}},
\qquad B=1,\dots,d,
\label{eq:LPH-vertical-Atiyah}
\\[0.4em]
\dot{z} &= L(q^i,\dot q^i,v^A,z).
\label{eq:LPH-contact-Atiyah}
\end{align}

These are the \emph{Lagrange--Poincaré--Herglotz equations} on the Atiyah algebroid.  
They coincide with the reduced contact Lagrange--Poincaré equations previously obtained in 
\cite{anahory2024euler}, \cite{simoes2024symmetry}, \cite{anahory2025contact}.
\end{example}

\begin{example}[Dissipative Wong's equations]\label{wong} To illustrate the above Lagrange--Poincaré--Herglotz equations, we consider a dissipative version of the classical Wong equations. These equations arise, for instance, in the dynamics of a charged particle moving in a Yang–Mills field and in the geometric analysis of the falling cat problem (see \cite{cendra2001lagrangian} and references therein). In both situations, the reduced dynamics is governed by the Atiyah algebroid associated with a principal $G$-bundle, and the internal variables evolve in the adjoint bundle $Ad(Q)$.

The Herglotz framework provides a natural geometric mechanism to incorporate dissipative effects in these models. A term linear in the contact variable $z$ produces, after reduction, a controlled decay of the horizontal and vertical momenta, in accordance with the notion of dissipated quantities introduced in \cite{de2021symmetries}. Physically, this allows one to model frictional or damping interactions with a surrounding medium (as in the case of a particle in a non-ideal Yang–Mills environment) and, in the context of the falling cat, the gradual loss of internal rotational energy due to elastic or viscoelastic couplings inside the body. In this way, Herglotz-type dissipation becomes compatible with the symmetry reduction and yields a gauge-invariant dissipative extension of Wong’s equations.

Let $(M, g_{M})$ be a Riemannian manifold, $G$ a compact Lie group endowed
with a bi-invariant Riemannian metric $\kappa$, and let $\pi : Q \to M$ be a
principal $G$-bundle with Lie algebra $\mathfrak{g}$. Let 
$\mathcal{A}: TQ \to {\mathfrak g}$ be a principal connection with curvature
$B:TQ\oplus TQ\to\mathfrak{g}$. Using $\mathcal{A}$ one gets an identification $T_qQ \;\cong\; T_{\pi(q)}M \oplus \mathfrak{g}$ with $q\in Q$, and the metrics $g_M$ and $\kappa$ jointly induce a $G$-invariant Riemannian
metric $g_Q$ on $Q$.

We consider the \emph{kinetic Herglotz Lagrangian}
$L:TQ\times\mathbb{R}\to\mathbb{R}$ given by
\begin{equation}\label{eq:Wong-unreduced-L}
L(v_q,z)
= \frac{1}{2}\,\Big(
\kappa_e\big(\mathcal{A}(v_q),\mathcal{A}(v_q)\big)
+ g_{M,\pi(q)}\big((T_q\pi)(v_q),(T_q\pi)(v_q)\big)
\Big) - \gamma\,z,
\end{equation}
for $v_q\in T_qQ$, where $e\in G$ is the identity and $\gamma>0$ is a
dissipation parameter. Clearly $L$ is hyperregular and $G$-invariant.

Since $g_Q$ is $G$-invariant, it induces a fiber metric $g_{TQ/G}$ on the
Atiyah algebroid $\widehat{TQ}=TQ/G\to M=Q/G$.  
The reduced Herglotz Lagrangian $L_{\mathrm{red}}\colon \widehat{TQ}\times\mathbb{R}\to\mathbb{R}$ is then the kinetic energy of $g_{TQ/G}$ minus the linear term in $z$:
\begin{equation}\label{eq:reduced-lagrangian-Wong}
L_{\mathrm{red}}([v_q],z)
= \frac{1}{2}\,\Big(
\kappa_e\big(\mathcal{A}(v_q),\mathcal{A}(v_q)\big)
+ g_{M,\pi(q)}\big((T_q\pi)(v_q),(T_q\pi)(v_q)\big)
\Big) - \gamma\,z,
\end{equation}
for $[v_q]\in TQ/G$. The Legendre transform associated with 
$L_{\mathrm{red}}$ is the bundle isomorphism $([v_q],s)\;\mapsto\;
\big(\flat_{g_{TQ/G}}([v_q]),s\big)$, where $\flat_{g_{TQ/G}}:TQ/G\to T^*Q/G$ is induced by the fiber metric $g_{TQ/G}$.

Let us choose a local trivialization $\pi^{-1}(U)\simeq U\times G$ of the
principal bundle $\pi:Q\to M$, where $U\subset M$ has local coordinates
$(q^i)$. Let $\{\xi_A\}$ be a basis of $\mathfrak{g}$ with structure
constants $c_{AB}^{\ \ D}$, let $\mathcal{A}_i^A$ and $\mathcal{B}_{ij}^A$
be the local components of $\mathcal{A}$ and its curvature $B$, and write $\kappa_e = \kappa_{AB}\,\xi^A\otimes\xi^B$ and $g_M = g_{ij}\,dq^i\otimes dq^j$, with $\{\xi^A\}$ the dual basis of $\{\xi_A\}$. The bi-invariance of $\kappa$
implies the identity $c_{AB}^{\ \ D}\,\kappa_{DE} = c_{AE}^{\ \ D}\,\kappa_{DB}$.

Denote by $\{e_i,\widehat{\xi}_A\}$ the associated $G$-invariant local frame
on $Q$ and by $(q^i,\dot q^i,v^A,z)$ the induced local fibred coordinates
on $\widehat{TQ}\times\mathbb{R}$. In these coordinates the reduced
Lagrangian \eqref{eq:reduced-lagrangian-Wong} becomes
\begin{equation}\label{eq:Wong-reduced-L}
L_{\mathrm{red}}(q^i,\dot q^i,v^A,z)
= \frac{1}{2}\,
\big(\kappa_{AB}v^A v^B + g_{ij}\dot q^i\dot q^j\big) - \gamma\,z.
\end{equation}
The Hessian of $L_{\mathrm{red}}$ with respect to $(\dot q^i,v^A)$ is $\displaystyle{W_{L_{\mathrm{red}}}
= \begin{pmatrix}
g_{ij} & 0\\[2pt]
0 & \kappa_{AB}
\end{pmatrix}}$, hence $L_{\mathrm{red}}$ is hyperregular.

Applying the Lagrange--Poincar\'e--Herglotz equations
\eqref{eq:LPH-horizontal-Atiyah}--\eqref{eq:LPH-contact-Atiyah} to
\eqref{eq:Wong-reduced-L}, we obtain the \emph{dissipative Wong equations}
on the base $M$:
\begin{align}
\frac{\partial g_{im}}{\partial q^{j}}\,\dot q^{i}\dot q^{m}
 - \frac{\partial g_{ij}}{\partial q^{k}}\,\dot q^{k}\dot q^{i}
 - g_{ij}\,\ddot q^{i}
&=
\kappa_{AB}\,v^{B}
\Big(\mathcal{B}_{ij}^{A}\,\dot q^{i}
 + c_{DB}^{\ \ A}\,\mathcal{A}_{j}^{B}\,v^{D}\Big)
 + \gamma\,g_{ij}\,\dot q^{i},
\label{eq:Wong-dissipative-horizontal}
\\[0.3em]
\kappa_{AB}\,\dot v^{A}
&=
\kappa_{AE}v^{E}
\Big(
c_{DB}^{\ \ A}\,v^{D}
 - c_{DB}^{\ \ A}\,\mathcal{A}_{i}^{D}\,\dot q^{i}
\Big)
 - \gamma\,\kappa_{AB}\,v^{A},
\label{eq:Wong-dissipative-vertical}
\end{align}
together with the contact equation $\dot z = L_{\mathrm{red}}(q^i,\dot q^i,v^A,z)$.

\end{example}

\section{Energy Balance and Noether--Herglotz Theorems}\label{section5}

In this section we study how the presence of a Herglotz (or contact-type)
term modifies the usual conservation laws associated with a Lagrangian
system on a Lie algebroid.  
While in the conservative setting, energy and Noether momenta are
constants of motion, here they satisfy first-order balance laws driven by
the derivative $\partial L/\partial z$ of the Herglotz Lagrangian with
respect to the contact variable.  We show that, after multiplication by a
canonical integrating factor, these quantities become time-independent
\emph{dissipated invariants} in the sense of \cite{bravetti2017contact}, \cite{gaset2020new}, \cite{de2021symmetries}.

\subsection{Energy balance law}

Next, we introduce the \emph{Herglotz energy}
associated with a Lagrangian $L\colon E\times\mathbb{R}\to\mathbb{R}$
and derive a general energy balance law along solutions of the
Euler--Lagrange--Herglotz equations.
We interpret the resulting rescaled conserved quantity and illustrate the
construction in the classical case $E=TQ$ by means of a mechanical
system with Herglotz-type Rayleigh dissipation, for which the mechanical
energy decays exponentially in time.

Given the coordinates $(x^{i},y^{\alpha},z)$ on $E\times\mathbb{R}$ adapted to a basis of sections of $E$, the \emph{Herglotz energy} is the smooth function on $E\times\mathbb{R}$ locally defined by
\begin{equation}\label{eq:energy_def}
E(x^{i},y^{\alpha},z) = \frac{\partial L}{\partial y^\alpha}(x^{i},y^{\alpha},z)\,y^\alpha - L(x^{i},y^{\alpha},z).
\end{equation}
The previous expression is the local version of the energy function as defined on \cite{anahory2024euler,anahory2025contact}.

\begin{proposition}[Energy balance law]
\label{prop:energy_balance}
Let $(x^{i}(t),y^{\alpha}(t),z(t))$ be a solution of the Herglotz equations
\eqref{eq:ELH_constraints}--\eqref{eq:ELH_equations}.  
Then, the energy $E$ satisfies
\begin{equation}\label{eq:energy_balance}
\dot E = \frac{\partial L}{\partial z}\,E.
\end{equation}
Equivalently,
\[
\frac{d}{dt}\left(
e^{-\int_0^t \frac{\partial L}{\partial z}(\omega)\,d\omega}\,E(t)
\right)=0
\]
along the trajectories.
\end{proposition}

\begin{proof}
By definition,
\[
E=\frac{\partial L}{\partial y^\alpha}y^\alpha - L.
\]
Differentiating and cancelling the terms 
$\frac{\partial L}{\partial y^\alpha}\dot y^\alpha$ gives
\[
\dot E 
= 
\frac{d}{dt}\left(\frac{\partial L}{\partial y^\alpha}\right)y^\alpha
- \frac{\partial L}{\partial x^i}\dot x^i
- \frac{\partial L}{\partial z}\dot z.
\]
Using $\dot x^i=\rho^i_{\ \alpha}y^\alpha$, we obtain
\[
\dot E
=
\left(
\frac{d}{dt}\left(\frac{\partial L}{\partial y^\alpha}\right)
-
\rho^i_{\ \alpha}\frac{\partial L}{\partial x^i}
\right)y^\alpha
-
\frac{\partial L}{\partial z}\dot z.
\]

Inserting the Euler--Lagrange--Herglotz equations in the previous equations we obtain
\[
\dot E
=
\left(
- C^\gamma_{\ \alpha\beta}y^\beta\frac{\partial L}{\partial y^\gamma}
+ \frac{\partial L}{\partial z}\frac{\partial L}{\partial y^\alpha}
\right)y^\alpha
-
\frac{\partial L}{\partial z}\dot z.
\]

Due to the skew-symmetry of the structure functions with respect to the lower indices, we have that 
$C^\gamma_{\ \alpha\beta}y^\alpha y^\beta=0$ which yields
\[
\dot E
=
\frac{\partial L}{\partial z}\frac{\partial L}{\partial y^\alpha}y^\alpha
-
\frac{\partial L}{\partial z}\dot z.
\]

Thus, using the Herglotz constraint $\dot z=L$ we get
\[
\dot E
=
\frac{\partial L}{\partial z} \left(\frac{\partial L}{\partial y^\alpha}y^\alpha - L\right) = \frac{\partial L}{\partial z}\,E.
\]
\end{proof}
\begin{example}[Mechanical system with Herglotz-type Rayleigh dissipation]
Consider the standard Lie algebroid $E = TQ$ over a configuration manifold $Q$,
with local coordinates $(q^i,\dot q^i)$ on $TQ$.
Let $g$ be a Riemannian metric on $Q$ and $V \colon Q \to \mathbb{R}$ a potential.
We define the Herglotz-type Lagrangian
\[
L(q,\dot q,z)
= \frac{1}{2}\,g_{ij}(q)\,\dot q^i \dot q^j - V(q) - \gamma\,z,
\qquad \gamma > 0.
\]
This is the usual kinetic minus potential energy, perturbed by a linear 
Herglotz term in the contact variable $z$, which models Rayleigh-type 
dissipation.

The generalized Herglotz energy is
\[
E(q,\dot q,z)
= \frac{\partial L}{\partial \dot q^{i}}(q,\dot q,z)\,\dot q^{i} - L(q,\dot q,z)
= \frac{1}{2}\,g_{ij}(q)\,\dot q^i \dot q^j + V(q) + \gamma z.
\]

From Proposition~\ref{prop:energy_balance} we obtain the balance equation
\[
\dot E = \frac{\partial L}{\partial z}\,E = -\gamma\,E,
\]
and therefore
\[
E(t) = E(0)\,e^{-\gamma t}.
\]
Equivalently, the dissipated quantity
\[
E_{\mathrm{diss}}(t) 
:= e^{-\int_0^t \frac{\partial L}{\partial z}(\omega)\,d\omega}\, E(t)
= e^{\gamma t} E(t)
\]
is conserved.  
In the terminology of \cite{de2021symmetries}, 
$E_{\mathrm{diss}}$ is a \emph{dissipated invariant}.

\end{example}

\subsection{Noether--Herglotz theorem on a Lie algebroid}

In the following, we develop a Noether--type theorem for
Herglotz systems on Lie algebroids.  Given an infinitesimal symmetry of
$L$, we construct the corresponding \emph{Herglotz momentum} and prove
a Noether--Herglotz balance law.
 The momentum is no longer conserved but obeys an equation of the same
form as the energy balance, and its rescaled version is
constant.  We then apply this result to several examples: a rigid body
on a Lie algebra with exponentially damped body angular momentum, a class
of dissipative Wong-type systems where both the kinetic energy and the
momentum become dissipated invariants, and a thermoviscous rigid
body formulated as an Euler--Poincar\'e--Herglotz system, in which the
Herglotz variable acquires a thermodynamic interpretation as an entropy.

Next, we will define a symmetry generated by a section $\sigma\in\Gamma(E)$. Recall that given a section $\sigma$, with local expression $\sigma(x^{i})=\sigma^{\alpha}(x^{i})e_{\alpha}$, we may consider two special sections in the prolongation Lie algebroid $\mathcal{T}^{E}(E\times \mathbb{R})\to E\times \mathbb{R}$: the vertical lift $\sigma^{V}$
$$\sigma^{V}(v_{q},z)=\left(0_{q},\left.\frac{d}{dt}\right|_{t=0}(v_{q}+t\sigma(q))\right),$$
whose local expression is
$$\sigma^{V}(x^{i},y^{\alpha},z)=\left(0_{q},\sigma^{\alpha}(x^{i})\frac{\partial}{\partial y^{\alpha}}\right),$$
and the complete lift $\sigma^{C}$ whose local expression is
$$\sigma^{C}(x^{i},y^{\alpha},z)=\left( \sigma^{\alpha}e_{\alpha}, \rho^{i}_\alpha \sigma^{\alpha}\frac{\partial}{\partial x^{i}} + \left(\rho^{i}_\alpha y^{\alpha}\frac{\partial \sigma^{\beta}}{\partial x^{i}} + C^\beta_{\ \alpha\gamma}y^\alpha\sigma^\gamma \right)\frac{\partial}{\partial y^{\beta}} \right)$$

\begin{definition}
We say that $\sigma\in\Gamma(E)$ is an \emph{infinitesimal symmetry}
of the Herglotz Lagrangian $L$ if the flow generated by $\sigma^{C}$ on $E\times\mathbb{R}$ leaves $L$ invariant. Infinitesimally, this means $\mathcal{L}_{\sigma^{C}}^E L = 0$, where $\mathcal{L}_{\sigma^{C}}^E$ is the Lie derivative act\-ing on functions on $E$ via the algebroid structure, i.e.,
$$\mathcal{L}_{\sigma^{C}}^E L = \rho_{1}(\sigma^{C})L,$$
for any function $L$ on $E\times \mathbb{R}$.
\end{definition}

In local coordinates, the previous condition simply means that if $\sigma$ is an infinitesimal symmetry then
\begin{equation}\label{local:infinitesimal:symmetries}
    \rho^{i}_\alpha \sigma^{\alpha}\frac{\partial L}{\partial x^{i}} + \rho^{i}_\alpha y^{\alpha}\frac{\partial \sigma^{\beta}}{\partial x^{i}}\frac{\partial L}{\partial y^{\beta}} + C^\beta_{\ \alpha\gamma}y^\alpha\sigma^\gamma \frac{\partial L}{\partial y^{\beta}} = 0.
\end{equation}

Define the \emph{Herglotz momentum} associated with $\sigma$ by the local expression
\[
J_\sigma = \frac{\partial L}{\partial y^\alpha}(x^{i},y^{\alpha},z)\,\sigma^\alpha(x^{i}).
\]

\begin{proposition}[Noether--Herglotz on a Lie algebroid]
\label{prop:noether_herglotz}
Let $\sigma\in\Gamma(E)$ be an infinitesimal symmetry of $L$.
Let $(x^{i}(t),y^{\alpha}(t),z(t))$ be a solution of the Herglotz equations.
Then the associated momentum $J_\sigma$ satisfies
\begin{equation}\label{eq:J_balance}
\dot J_\sigma = \frac{\partial L}{\partial z}\,J_\sigma.
\end{equation}
Equivalently,
\[
\frac{d}{dt}\left(
e^{-\int_0^t \frac{\partial L}{\partial z}(\omega)\,d\omega}\,J_\sigma(t)
\right) = 0
\]
along the trajectories.
\end{proposition}

Hence the rescaled momentum 
$e^{-\int_0^t (\partial L/\partial z)(\omega)\,d\omega}\,J_\sigma$
is a time-dependent quantity conserved in the classical sense.

\begin{proof}
By definition,
\[
J_\sigma = \frac{\partial L}{\partial y^\alpha}\,\sigma^\alpha(x).
\]
Differentiate with respect to $t$:
\[
\dot J_\sigma
= \frac{d}{dt}\left(\frac{\partial L}{\partial y^\alpha}\right)\sigma^\alpha
+ \frac{\partial L}{\partial y^\alpha}\frac{d}{dt}\big(\sigma^\alpha(x)\big).
\]
Using the chain rule,
\[
\frac{d}{dt}(\sigma^\alpha(x)) = 
\frac{\partial \sigma^\alpha}{\partial x^i} \dot x^i
= \frac{\partial \sigma^\alpha}{\partial x^i}\,\rho^i_{\ \beta}y^\beta.
\]
and substituting $\dot x^i=\rho^i_{\ \beta}y^\beta$. Inserting the Euler--Lagrange--Herglotz equations, we obtain
\begin{align*}
\dot J_\sigma
&= \left(
\rho^i_{\ \alpha}\frac{\partial L}{\partial x^i}
- C^\gamma_{\ \alpha\beta}y^\beta \frac{\partial L}{\partial y^\gamma}
+ \frac{\partial L}{\partial z}\frac{\partial L}{\partial y^\alpha}
\right)\sigma^\alpha
+ \frac{\partial L}{\partial y^\alpha}\frac{\partial \sigma^\alpha}{\partial x^i}
\rho^i_{\ \beta}y^\beta.
\end{align*}

Finally, the local expression \eqref{local:infinitesimal:symmetries} satisfied by $\sigma$ due to the fact that it is an infinitesimal symmetry, implies that almost all the terms are zero except for the term containing the partial derivative with respect to the variable $z$ which results in
\[
\dot J_\sigma = \frac{\partial L}{\partial z}\,J_\sigma.
\]
The integrating factor argument then shows that
$e^{-\int_0^t (\partial L/\partial z)(\omega)\,d\omega} J_\sigma(t)$
is constant.
\end{proof}

\begin{remark}
In the conservative case $\partial L/\partial z=0$, we recover the usual
Noether theorem: $J_\sigma$ is conserved.\hfill$\diamond$ 
\end{remark}

\begin{example}[Rigid body with dissipated body angular momentum]
We next consider a rigid body with configuration Lie group $G = \mathrm{SO}(3)$
and Lie algebra $\mathfrak{g} \simeq \mathbb{R}^3$.
We work on the Lie algebra viewed as a Lie algebroid over a point, $E=\mathfrak{g}$.
Let $\mathbb{I} \colon \mathfrak{g} \to \mathfrak{g}^*$ be the inertia operator.
We choose body angular velocity $\xi \in \mathfrak{g}$ and define for $\gamma>0$ the
Herglotz-type Lagrangian (see \cite{anahory2024reduction})
\[
\ell(\xi,z)
= \frac{1}{2}\,\langle \mathbb{I}\xi,\xi\rangle - \gamma\,z.
\]

The Euler--Poincaré--Herglotz equations on $\mathfrak{g}$ read
\[
\frac{d}{dt}\left(\frac{\partial \ell}{\partial \xi}\right)
+ \operatorname{ad}^*_{\xi}\!\left(\frac{\partial \ell}{\partial \xi}\right)
= \frac{\partial \ell}{\partial z}\,\frac{\partial \ell}{\partial \xi},
\qquad \dot z = \ell(\xi,z).
\]
Identifying the body angular momentum $\mu := \partial \ell/\partial \xi
= \mathbb{I}\xi$, we obtain
\[
\dot\mu + \operatorname{ad}^*_{\xi}\mu = -\gamma\,\mu.
\]

Suppose now that the inertia tensor is axially symmetric about the body $e_3$-axis,
so that the Lagrangian is invariant under rotations generated by $e_3$.
In the conservative case ($\gamma=0$), Noether's theorem implies that the
component $\langle \mu, e_3\rangle$ of the body angular momentum is conserved (see \cite{marsden2013introduction} for instance).

In the Herglotz setting, Noether--Herglotz Proposition~\ref{prop:noether_herglotz}
gives instead
\[
\frac{d}{dt} J_\sigma = \frac{\partial \ell}{\partial z}\,J_\sigma
= -\gamma\,J_\sigma,
\]
where $J_\sigma(t) = \langle \mu(t),e_3\rangle$ is the momentum associated with
the symmetry $\sigma=e_3$.
Thus
\[
J_\sigma(t) = J_\sigma(0)\,e^{-\gamma t}
\quad\text{and}\quad
\widetilde J_\sigma(t) := e^{\gamma t}\,J_\sigma(t)
\ \text{is conserved}.
\]

From the viewpoint of \cite{de2021symmetries}, $\widetilde J_\sigma$ is a
\emph{dissipated quantity}: the usual body component of angular momentum
decays exponentially, but after multiplication by the integrating factor
determined by the Herglotz term it becomes an invariant of the motion.
\end{example}

\begin{example}[Dissipated energy and momenta in Wong-type systems]
We finally revisit the dissipative Wong equations of
Example~\ref{wong} from the viewpoint of
energy balance and dissipated quantities for particles with  symmetries.

Let $L_{\mathrm{red}} \colon \widehat{TQ}\times\mathbb{R}\to\mathbb{R}$
be the reduced Herglotz Lagrangian on the Atiyah algebroid
$\widehat{TQ}=TQ/G \to M$, of the form
\[
L_{\mathrm{red}}(q^i,\dot q^i,v^A,z)
= \frac{1}{2}\big(
\kappa_{AB}\,v^A v^B + g_{ij}\,\dot q^i \dot q^j\big) - \gamma\,z,
\qquad \gamma>0,
\]
where $g_{ij}$ and $\kappa_{AB}$ encode the metrics on the base $M$ and the
Lie algebra $\mathfrak{g}$, respectively.

The corresponding generalized energy is
\[
E(q,\dot q,v,z)
= \left\langle \frac{\partial L_{\mathrm{red}}}{\partial \dot q},\,\dot q\right\rangle
+ \left\langle \frac{\partial L_{\mathrm{red}}}{\partial v},\,v\right\rangle
- L_{\mathrm{red}}(q,\dot q,v,z)
= \frac{1}{2}\big(
\kappa_{AB}\,v^A v^B + g_{ij}\,\dot q^i \dot q^j\big) + \gamma z,
\]
that is, the total kinetic energy. Since $\partial L_{\mathrm{red}}/\partial z = -\gamma$,
Proposition~\ref{prop:energy_balance} implies
\[
\dot E = \frac{\partial L_{\mathrm{red}}}{\partial z}\,E = -\gamma\,E,
\]
and hence
\[
E(t) = E(0)\,e^{-\gamma t},
\qquad
\widetilde E(t) := e^{\gamma t}\,E(t)
\ \text{is conserved}.
\]
Thus the kinetic energy of the Wong system decays exponentially, while the
rescaled quantity $\widetilde E$ is a dissipated invariant.

A second family of dissipated quantities arises from the internal
symmetry. Consider the $G$-action on $Q$ and the induced $G$-action on
$\widehat{TQ}$, and let $\sigma\in\Gamma(E)$ be the section corresponding
to a fixed element $\eta\in\mathfrak{g}$. In the non-dissipative case,
Noether's theorem yields conservation of the associated momentum
$\langle \frac{\partial L_{\mathrm{red}}}{\partial v},\eta\rangle$ \cite{cendra2001lagrangian}.
In the Herglotz case, Proposition~\ref{prop:noether_herglotz} gives
\[
\frac{d}{dt} J_\eta
= \frac{\partial L_{\mathrm{red}}}{\partial z}\,J_\eta
= -\gamma\,J_\eta,
\qquad
J_\eta(t) := \left\langle \frac{\partial L_{\mathrm{red}}}{\partial v},
\eta\right\rangle = \kappa_{AB}\,v^A \eta^B,
\]
and hence
\[
J_\eta(t) = J_\eta(0)\,e^{-\gamma t},
\qquad
\widetilde J_\eta(t) := e^{\gamma t}\,J_\eta(t)
\ \text{is conserved}.
\]

Therefore, both the kinetic energy and the momenta of the Wong system
become \emph{dissipated quantities}: they decay exponentially along solutions,
while their appropriately rescaled versions are genuine invariants of the
dissipative dynamics on the Atiyah algebroid.
\end{example}

\begin{example}[Thermoviscous rigid body as an Euler--Poincar\'e--Herglotz system]
\label{ex:EPH-thermoviscous-rigid-body}

We consider a simple finite-dimensional thermomechanical system:
a rigid body with configuration space $Q = SO(3)$, whose dynamics is affected by a linear viscous torque and coupled to an
entropy variable $S(t)$. The Herglotz variable $z$ will be interpreted
as the entropy $S$.

The Euler–Poincaré–Herglotz formulation of a thermoviscous rigid body provides a symmetry–reduced counterpart of the finite-dimensional thermodynamic systems with linear dissipation studied in Gay-Balmaz and Yoshimura \cite{GayBalmazYoshimura2017PartI}, and it parallels the rigid-body thermomechanical models on Lie groups of Couéraud and Gay-Balmaz \cite{CoueraudGayBalmaz2020}, where viscous torques and entropy production give rise to dissipative Euler-Poincaré-Herglotz dynamics.

Let $G=SO(3)$ act on $Q=SO(3)$ by left multiplication.
The associated action Lie algebroid is $E = Q\times\mathfrak{so}(3) \;\to\; Q,
\,\, \tau(R,\xi) = R$. Its structure is:
\begin{itemize}
  \item The anchor map $\rho \colon Q\times\mathfrak{so}(3) \to TQ,\,
  \rho(R,\xi) = \xi_Q(R)$, where $\xi_Q$ is the infinitesimal generator of the left action:
  $\xi_Q(R) = \frac{d}{dt}\big|_{t=0} \big( \exp(t\xi)\,R\big)$.
  \item The bracket on sections is induced from the Lie bracket on
  $\mathfrak{so}(3)$:
  \[
  [\sigma_1,\sigma_2](R) = \big(R,[\xi_1,\xi_2]\big),
  \quad \text{if }\sigma_i(R)=(R,\xi_i),\ i=1,2.
  \]
\end{itemize}

Choose a basis $\{e_\alpha\}_{\alpha=1}^3$ of $\mathfrak{so}(3)$ with
structure constants $[e_\alpha,e_\beta] = C_{\alpha\beta}^{\ \ \gamma} e_\gamma$, and write any $\xi\in\mathfrak{so}(3)$ as $\xi=\xi^\alpha e_\alpha$.
A curve in the algebroid is then $(R(t),\xi(t))$, with coordinates
$\xi^\alpha(t)$. The admissibility condition $\dot R = \rho(R,\xi)$ reads $\dot R(t) = \xi_Q\big(R(t)\big)
           = R(t)\,\widehat{\xi(t)}$, where $\widehat{\xi}\in \mathfrak{so}(3)$ is the skew-symmetric matrix
corresponding to $\xi\in\mathbb{R}^3$. Thus admissible curves encode the
usual rigid body kinematics in body coordinates.

Let $I\colon \mathfrak{so}(3)\to\mathfrak{so}(3)^*$ be a positive definite
inertia operator, and let $U\colon Q\to\mathbb{R}$ be a potential energy
(e.g.\ gravitational). We interpret the scalar variable $z(t)$ as the
entropy $S(t)$ of the body. Fix a constant temperature parameter $T_0>0$,
and a viscous coefficient $\gamma>0$.

We define the Herglotz-type Lagrangian
\[
L\colon E\times\mathbb{R}\to\mathbb{R},\qquad
L(R,\xi,S)
= \frac{1}{2}\,\langle I\xi,\xi\rangle - U(R)
  - T_0\,S \;+\; \frac{\gamma}{2T_0}\,\lVert\xi\rVert^2.
\]
Note that the term $\frac{1}{2}\langle I\xi,\xi\rangle - U(R)$ is the usual mechanical
  Lagrangian of the rigid body, $-T_0 S$ encodes the coupling to entropy at temperature $T_0$
  (Legendre-type term $-TS$), and the term $\frac{\gamma}{2T_0}\|\xi\|^2$ is a simple quadratic
  dissipation potential, scaled by $1/T_0$ so that the corresponding
  entropy production will be proportional to the dissipation rate.

The entropy satisfies $\dot S(t) = L\big(R(t),\xi(t),S(t)\big)$. For a Herglotz Lagrangian
\[
\ell(\xi,S) = L(R,\xi,S)
\quad\text{(here $\ell$ does not depend explicitly on $R$ if $U$ is left-invariant),}
\]
the Euler--Poincar\'e--Herglotz equations on the Lie algebra
$\mathfrak{so}(3)$ read (cf.\ Example~\ref{ex:EPH_action_intrinsic}
with trivial base dependence)
\begin{equation}\label{eq:EPH-rigid-body}
\frac{d}{dt}\left(\frac{\partial \ell}{\partial \xi}\right)
+ \operatorname{ad}^*_{\xi}\!\left(\frac{\partial \ell}{\partial \xi}\right)
= \frac{\partial \ell}{\partial S}\,
  \frac{\partial \ell}{\partial \xi},
\qquad
\dot S = \ell(\xi,S),
\end{equation}
together with the kinematic relation
\[
\dot R = R\widehat{\xi}.
\]

In our case,
\[
\frac{\partial \ell}{\partial \xi}
= I\xi + \frac{\gamma}{T_0}\,\xi
=: \Pi(\xi),
\qquad
\frac{\partial \ell}{\partial S} = -T_0.
\]
Thus \eqref{eq:EPH-rigid-body} becomes
\begin{equation}\label{eq:EPH-rigid-body-explicit}
\frac{d}{dt}\Pi(\xi)
+ \operatorname{ad}^*_{\xi}\Pi(\xi)
= -\,T_0\,\Pi(\xi),
\qquad
\dot S
= \frac{1}{2}\langle I\xi,\xi\rangle - U(R)
  - T_0\,S + \frac{\gamma}{2T_0}\lVert\xi\rVert^2.
\end{equation}
The first equation is an \emph{Euler--Poincar\'e equation with linear
damping} on the momentum $\Pi(\xi)$, while the second encodes the entropy
production.

The Herglotz energy
\[
E(R,\xi,S) =
\left\langle \frac{\partial L}{\partial \xi},\,\xi\right\rangle - L
= \langle I\xi,\xi\rangle + \frac{\gamma}{T_0}\lVert\xi\rVert^2
  - \bigg(\frac{1}{2}\langle I\xi,\xi\rangle - U(R)
         -T_0 S + \frac{\gamma}{2T_0}\lVert\xi\rVert^2\bigg)
\]
simplifies to
\[
E(R,\xi,S)
= \frac{1}{2}\langle I\xi,\xi\rangle
  + \frac{\gamma}{2T_0}\lVert\xi\rVert^2
  + U(R) + T_0 S.
\]
By the general energy balance law for Herglotz systems,
\[
\dot E = \frac{\partial L}{\partial S}\,E = -T_0\,E,
\]
so $E(t)$ decays exponentially:
\[
E(t) = E(0)\,e^{-T_0 t}.
\]

Energy decrease is not violating the first law of Thermodynamics since we are describing a non-isolated physical system loosing energy to its surrondings due to the dissipative term. At the same time, the entropy production law
\[
\dot S
= \frac{1}{2}\langle I\xi,\xi\rangle
  - U(R) - T_0 S
  + \frac{\gamma}{2T_0}\lVert\xi\rVert^2
\]
shows that, for suitable regimes (e.g.\ $U$ bounded below, large damping
or near equilibrium), the entropy $S$ increases, and its growth is
driven by the dissipative term $\frac{\gamma}{2T_0}\|\xi\|^2$.
In this sense, the Herglotz term realizes a finite-dimensional,
symmetry-reduced analogue of a thermodynamic entropy balance:
mechanical energy is dissipated and converted into entropy at a rate
encoded geometrically by the Herglotz Lagrangian.

\end{example}

\begin{remark}[Comparison with the heavy top in a Stokes flow]
The thermoviscous rigid body considered above is closely related to the
thermomechanical rigid-body models of Couéraud and Gay--Balmaz,
notably the heavy top in a Stokes flow \cite{CoueraudGayBalmaz2020}.
There, the Lagrangian consists of a mechanical part and an internal
energy $U_B(S)$ depending only on the entropy, while viscous effects
enter through phenomenological constraints that produce a torque of
Stokes type and an associated entropy-production law. The Lagrangian
itself encodes only the reversible part of the dynamics.

In contrast, the Herglotz formulation used here incorporates a simple
Rayleigh dissipation potential directly into the variational principle,
through the term $\tfrac{\gamma}{2T_0}\|\xi\|^2$ in $L(R,\xi,S)$.  The
internal contribution is still purely entropic (here $T_0S$, the
isothermal case of the general $U_B(S)$), but the velocity-dependent
dissipation is now encoded geometrically via the Herglotz term.  As a
result, the Euler--Poincaré--Herglotz equation contains a multiplicative
damping term $-T_0\Pi(\xi)$, and the associated Herglotz energy decays
exponentially rather than satisfying a conservation law for the sum of
mechanical and internal energies.

Thus the present model provides a fully variational, symmetry-reduced
counterpart of the thermomechanical heavy-top dynamics, with dissipation
implemented intrinsically through the Herglotz formalism rather than by
external constraints.\hfill$\diamond$
\end{remark}

\section{Coordinate-free formulation of the Euler--Lagrange--Herglotz equations}\label{sec:coordinate_free}


In order to obtain a coordinate-free formulation of the Herglotz equations,
we briefly recall Lie algebroid connections and the
Euler--Lagrange--Poincar\'e formalism developed in 
\cite{li2017lagrangian} and \cite{hu2024hamiltonian}.

Let $(E\to M,[\cdot,\cdot],\rho)$ be a Lie algebroid.
A \emph{$TM$-connection} on $E$ is a linear connection
\[
\nabla \colon \mathfrak{X}(M)\times\Gamma(E)\longrightarrow\Gamma(E),
\qquad (X,u)\mapsto\nabla_X u,
\]
on the vector bundle $E\to M$.  
It induces two $E$-connection on $E$, denoted again by $\nabla$ and $\overline{\nabla}$
\[
        \nabla_a u := \nabla_{\rho(a)}u, \qquad
        \overline{\nabla}_{a} u := \nabla_{\rho(u)}a + [a, u],
        \]
        for $a\in E$, $u\in\Gamma(E)$ . In addition, for each $E$-connection, a dual connection on the dual algebroid $\pi:E^{*}\to M$ denoted by $\nabla^*$ is induced satisfying
        \[
        \langle \nabla^*_a p, u \rangle := \rho(a)\left[\langle p,u \rangle \right]-\langle p,\nabla_{a}u \rangle \quad \forall u\in \Gamma(E).
        \]
        for $a\in E$ and $p\in\Gamma(E^*)$.

Given an admissible curve $a(t)\in E$ with base curve $x(t)\in M$,
we write $a(t)\in E_{x(t)}$ and denote by
\[
\nabla_{a(t)} \colon E_{x(t)}\longrightarrow E_{x(t)}
\]
the covariant derivative along $a(t)$ induced by the $E$-connection $\nabla$.
For a curve $u(t)\in E_{x(t)}$ over the same base curve $x(t)$ whose values coincide with those of a time-dependent section $\eta(t,x(t))$ of $\Gamma(E)$, i.e., $u(t)=\eta(t,x(t))$, we set
\[
\nabla_{a(t)}u(t):=\nabla_{a(t)} \eta (t,x(t)) + \dot{\eta}(t,x(t))\in E_{x(t)},
\]
where the covariant derivative is taken with respect to the $E$–connection.




Using a $TM$-connection on $E$ each vector $X\in T_{a}E$ can be decomposed in a pair $(X_{\mathrm{hor}},X_{\mathrm{ver}})$ with $X_{\mathrm{hor}}\in TM$ and $X_{\mathrm{ver}}\in E$, called the horizontal and vertical components of $X$. The vertical component belongs to the tangent to the fibers $X_{\mathrm{ver}}\in T_{a}E_{x}\approx E_{x}$, while the horizontal component is in one-to-one correspondence with $TM$ under the connection $\nabla$ (see Section 3.2 in \cite{LOJAFERNANDES2002119}).  Similarly, a $TM$-connection determines a splitting of $T^{*}E$ into horizontal and vertical subbundles. Any $\theta \in T^{*}_{a}E$ decomposes, relative to the $TM$–connection $\nabla$, in the pair $(\theta_{\mathrm{hor}},\theta_{\mathrm{ver}})$, with $\theta_{\mathrm{ver}}(a)\in E^*_{\tau(a)}$ and $\theta_{\mathrm{hor}}(a)\in T^*_{\tau(a)}M$ defined by
$$\langle\theta_{\mathrm{hor}},X \rangle = \langle\theta_{\mathrm{hor}},X_{\mathrm{hor}} \rangle, \qquad \langle\theta_{\mathrm{ver}},X \rangle = \langle\theta_{\mathrm{ver}},X_{\mathrm{ver}} \rangle,$$
for all $X\in T_{a}E$.

Following \cite{li2017lagrangian, hu2024hamiltonian}, given an admissible path $a(t)\in E$ over $x(t)\in M$ from $[0,1]$, an admissible variation is of the form $X_{b,a}(t)\in T_{a(t)}E$, where $b$ is a curve in $E$ over $x(t)$ but not necessarily admissible such that $b(0)=b(1)=0.$ Relative to a $TM$-connection, the vertical component of $X_{b,a}$ is $\overline{\nabla}_{a} b$ and the horizontal component is $\rho(b)$. In \cite{CrainicFernandes2003Integrability} it is shown that the space of admissible variations is independent of the choice of connection in the above definition.

When $L\colon E\to\mathbb{R}$ is an ordinary (non-Herglotz) Lagrangian, 
the Euler--Lagrange--Herglotz equations can be written in the 
connection-based form
\begin{equation}\label{eq:ELP_classical}
\overline{\nabla}^*_{a(t)}p(t)
-
\rho^*\!\big(dL_{\mathrm{hor}}(a(t))\big)
=0,
\qquad
p=dL_{\mathrm{ver}}(a),
\end{equation}
as shown in \cite[Section~3]{li2017lagrangian} and 
\cite[Theorem~5.3]{hu2024hamiltonian}. In \cite{li2017lagrangian} the notation $\overline{\nabla}^*$ does not denote the dual connection but the adoint operator, causing a difference in the sign that we see in the equations.

\subsection{Euler--Lagrange--Herglotz equations}

We now derive an intrinsic, connection-based formulation of the
Euler--Lagrange--Herglotz equations on a Lie algebroid,
in the spirit of \cite{li2017lagrangian}.%

\begin{theorem}[Euler--Lagrange--Herglotz equations]
\label{thm:ELP_Herglotz}
Let $(E\to M,[\cdot,\cdot],\rho)$ be a Lie algebroid equipped with a 
$TM$--connection $\nabla$, and let $\nabla^*$ be the induced $E$--connection
on $E^*$.  
Let $L\colon E\times\mathbb{R}\to\mathbb{R}$ be a Herglotz-type Lagrangian and let
$(x(t),a(t),z(t))$ be an admissible curve satisfying
\[
\dot x(t)=\rho(a(t)),
\qquad
\dot z(t)=L\big(a(t),z(t)\big).
\]
Decompose the differential of $L$ as
\[
dL(a,z)=dL_{\mathrm{hor}}(a,z)+dL_{\mathrm{ver}}(a,z)
+\frac{\partial L}{\partial z}(a,z)\,dz,
\]
and define the momentum
\[
p(t):=dL_{\mathrm{ver}}\big(a(t),z(t)\big)\in E^*_{x(t)}.
\]

Then an admissible curve $(x(t),a(t),z(t))$ satisfies the Euler--Lagrange--Herglotz
equations if and only if it satisfies the intrinsic relation
\begin{equation}\label{eq:ELP_Herglotz_intrinsic_short}
\overline{\nabla}^*_{a(t)}p(t)
-
\rho^*\!\big(dL_{\mathrm{hor}}(a(t),z(t))\big)
=\frac{\partial L}{\partial z}\big(a(t),z(t)\big)\,p(t),
\end{equation}
together with the Herglotz equation $\dot z=L(a,z)$.
\end{theorem}

\begin{proof}


Following the argument in
\cite{li2017lagrangian,hu2024hamiltonian} adapted to the contact case, take a variation $(X_{b,a}, \delta z)$ of a curve $(a(t), z(t))$ satisfying the Herglotz equation $\dot{z}=L(a(t),z(t))$. Then,
$$\frac{d}{dt}(\delta z)= \langle d L, (X_{b,a}, \delta z) \rangle = \langle dL_{\mathrm{hor}}(a), \rho(b) \rangle + \langle dL_{\mathrm{ver}}(a), \overline{\nabla}_{a} b \rangle + \frac{\partial L}{\partial z}\delta z.$$
Solving for $\delta z(t)$ using the integrating factor
\[
\lambda(t) := \exp\!\left(-\int_0^t \frac{\partial L}{\partial z}(\omega)\,d\omega\right),
\]
and also that $\delta z(0)=0$ and $\lambda(t)\neq 0$, the stationarity condition $\delta z(T)=0$ is equivalent to

\begin{equation*}
\int_0^T \lambda(t)
\left(
\langle dL_{\mathrm{hor}}(a), \rho(b) \rangle + \langle dL_{\mathrm{ver}}(a), \overline{\nabla}_{a} b \rangle
\right) \,dt = 0.
\end{equation*}
Equivalently, after integrating by parts the therm with $\overline{\nabla}_{a} b$ and using the fact that $b$ vanishes in the boundaries, we obtain
\[
\int_0^T \left(
\langle \lambda(t)\rho^{*}dL_{\mathrm{hor}}(a), b \rangle - \langle \overline{\nabla}_{a}^{*}\left(\lambda(t)dL_{\mathrm{ver}}(a)\right),  b \rangle
\right) \,dt = 0.
\]
The second term in the integrand satisfies
$$\overline{\nabla}_{a}^{*}\left(\lambda(t)dL_{\mathrm{ver}}(a)\right) = \lambda(t)\overline{\nabla}_{a}^{*}\left(dL_{\mathrm{ver}}(a)\right)  +\dot{\lambda}(t) dL_{\mathrm{ver}}(a).$$
Hence, we obtain
\[
\int_0^T \left(
\lambda(t)\langle \rho^{*}dL_{\mathrm{hor}}(a) -\overline{\nabla}_{a}^{*}\left(dL_{\mathrm{ver}}(a)\right) + \frac{\partial L}{\partial z} dL_{\mathrm{ver}}(a),  b \rangle
\right) \,dt = 0.
\]
for arbitrary variations $b$. Thus, from the fundamental theorem of calculus of variations we must have
$$\rho^{*}dL_{\mathrm{hor}}(a) - \overline{\nabla}_{a}^{*}\left(dL_{\mathrm{ver}}(a)\right) + \frac{\partial L}{\partial z} dL_{\mathrm{ver}}(a) = 0.$$
The converse follows by reversing the above computations, showing that then
the integral vanishes for all admissible variations.
Alternatively, expanding \eqref{eq:ELP_Herglotz_intrinsic_short} in a local
frame (see Lemma~\ref{lem:local_ELP_Herglotz} in the Appendix) yields the
local Euler--Lagrange--Herglotz equations.  
\end{proof}

\begin{remark}
    An important remark is that the local expression of the equations of motion is independent of the choice of $TM$-connection in the previous theorem, as it is pointed out in \cite{hu2024hamiltonian} (see also Lemma~\ref{lem:local_ELP_Herglotz} in the Appendix).
\end{remark}

\section{Hamilton--Pontryagin--Herglotz principle on a Lie algebroid}
\label{sec:HPH}

In the Hamilton--Pontryagin formulation of Lagrangian mechanics on a Lie
algebroid \cite{li2017lagrangian}, one considers triples of paths
$(a,v,p)$ in $E$ and $E^*$ together with a base path $x(t)$, and imposes
the kinematic constraint $a=v$ through a Lagrange multiplier $p$.
In the Herglotz setting, the action is replaced by a scalar variable
$z(t)$ satisfying a contact-type evolution equation.

Following \cite[Def.~5.3]{li2017lagrangian}, an $(E,E^*)$--path consists of:
\begin{enumerate}
\item[(i)] an $E$–path $a$ over a base curve $x\colon I\to M$;
\item[(ii)] an arbitrary curve $v$ on $E$ over the same base curve $x$ but not necessarily an $E$-path;
\item[(iii)] a curve $p$ on $E^*$ over $x$.
\end{enumerate}

Admissible variations of $(a,v,p, z)$ are the same as in
\cite{li2017lagrangian}: $\delta a=X_{b,a}$, $\delta z$, $\delta v$ and $\delta p$ are any variations satisfying $\tau_{*}(\delta v)=\pi_{*}(\delta p)=\rho(b)$ 



Let $L\colon E\times\mathbb{R}\to \mathbb{R}$ be a Herglotz Lagrangian.
Instead of an action functional, we define the
\emph{Pontryagin--Herglotz action variable} $z(t)$ by
\begin{equation}\label{eq:HPH_z_revised_final}
\dot z(t)
=
L\big(v(t),z(t)\big)
+
\langle p(t),\,a(t)-v(t)\rangle.
\end{equation}

Given a 1–parameter family $(a_s,v_s,p_s,z_s)$ of variations, the
\emph{Hamilton--Pontryagin--Herglotz principle} is
\[
\left.\frac{d}{ds}\right|_{s=0} z_s(T)=0,
\]
where $z_s$ solves the perturbed equation obtained from
\eqref{eq:HPH_z_revised_final}.

\begin{definition}
A curve $(a,v,p,z)$ satisfies the
\emph{Hamilton--Pontryagin--Herglotz principle}
if for every admissible variation
$(\delta a,\delta v,\delta p,\delta z)$ with
$\delta x(0)=\delta x(T)=0$ and $\delta z(0)=0$ one has
$\displaystyle \left.\frac{d}{ds}\right|_{s=0} z_s(T)=0$.
\end{definition}

\begin{theorem}[Implicit Euler--Lagrange--Poincar\'e--Herglotz equations]
\label{thm:implicit_ELP_Herglotz}
If $(a,v,p,z)$ satisfies the Hamilton--Pontryagin--Herglotz principle,
then:
\begin{equation}\label{eq:implicit_HPH_final_corrected}
\begin{aligned}
&\text{\em (i)}\quad a(t)\ \text{is an $E$–path: } \dot x=\rho(a),\\[0.2em]
&\text{\em (ii)}\quad a(t)=v(t),\\[0.2em]
&\text{\em (iii)}\quad p(t)=dL_{\mathrm{ver}}\big(a(t),z(t)\big),\\[0.2em]
&\text{\em (iv)}\quad
\overline{\nabla}_{a}^{*}p (t) - \rho^{*}dL_{\mathrm{hor}}(a(t),z(t)) - \frac{\partial L}{\partial z}(a(t),z(t)) p(t) = 0,\\[0.2em]
&\text{\em (v)}\quad \dot z(t)=L(a(t),z(t)).
\end{aligned}
\end{equation}
Conversely, any curve satisfying (i)--(v) satisfies the
Hamilton--Pontryagin--Herglotz principle and projects to a Herglotz
extremal.
\end{theorem}

\begin{proof}
Take an admissible variation $(X_{b,a}, \delta v, \delta p, \delta z)$ of a curve $(a(t), v(t), p(t), z(t))$ satisfying the Herglotz equation $\dot{z}=L(a(t),z(t))+
\langle p(t),\,a(t)-v(t)\rangle$. Then,
$$\frac{d}{dt}(\delta z)= \langle d L, (\delta v, \delta z) \rangle + \langle \delta p, a -v \rangle + \langle p, X_{b,a}\rangle - \langle p,\delta v\rangle= \langle dL_{\mathrm{hor}}(a), \rho(b) \rangle + \langle dL_{\mathrm{ver}}(a), \delta v_{\mathrm{ver}} \rangle + \frac{\partial L}{\partial z}\delta z$$

$$+\langle \delta p_{\mathrm{ver}}, a-v \rangle + \langle \rho(b), a-v\rangle + \langle p, \overline{\nabla}_{a} b - \delta v_{\mathrm{ver}} \rangle.$$

Solving for $\delta z(t)$ using the integrating factor
\[
\lambda(t) := \exp\!\left(-\int_0^t \frac{\partial L}{\partial z}(\omega)\,d\omega\right),
\]
and also that $\delta z(0)=0$ and $\lambda(t)\neq 0$, the stationarity condition $\delta z(T)=0$ is equivalent to

\begin{equation*}
\int_0^T \lambda(t)
\left(
\langle dL_{\mathrm{hor}}, \rho(b) \rangle + \langle dL_{\mathrm{ver}},  \delta v_{\mathrm{ver}}\rangle +\langle \delta p_{\mathrm{ver}}, a-v \rangle + \langle \rho(b), a-v\rangle + \langle p, \overline{\nabla}_{a} b - \delta v_{\mathrm{ver}} \rangle
\right) \,dt = 0.
\end{equation*}
Equivalently,
\[
\int_0^T \left(
\langle \lambda(t)\rho^{*}dL_{\mathrm{hor}}(a) + \lambda(t)\rho^{*}(a-v) -  \overline{\nabla}_{a}^{*}\left(\lambda(t)p\right), b \rangle + \langle dL_{\mathrm{ver}} - p,  \delta v_{\mathrm{ver}}\rangle + \langle \delta p_{\mathrm{ver}}, a-v \rangle
\right) \,dt = 0.
\]
The third term in the integrand satisfies
$$\overline{\nabla}_{a}^{*}\left(\lambda(t)p(t)\right) = \lambda(t)\overline{\nabla}_{a}^{*} p(t)  +\dot{\lambda}(t) p(t).$$
Hence, we obtain
\[
\int_0^T \left(
\lambda(t)\langle \rho^{*}dL_{\mathrm{hor}}(a) + \rho^{*}(a-v) - \overline{\nabla}_{a}^{*} p  + \frac{\partial L}{\partial z} p, b \rangle + \langle dL_{\mathrm{ver}} - p,  \delta v_{\mathrm{ver}}\rangle + \langle \delta p_{\mathrm{ver}}, a-v \rangle
\right) \,dt = 0.
\]
for arbitrary variations $b$, $\delta v_{\mathrm{ver}}$ and $\delta p_{\mathrm{ver}}$. Thus, from the fundamental theorem of calculus of variations we must have
$$\rho^{*}dL_{\mathrm{hor}}(a) - \overline{\nabla}_{a}^{*}p + \frac{\partial L}{\partial z} p = 0, \quad dL_{\mathrm{ver}} = p, \quad a=v.$$

Conversely, if we reverse all the arguments, the integrand in $\delta z(T)$ vanishes
identically, so the Hamilton--Pontryagin--Herglotz principle holds.
\end{proof}

\begin{corollary}[Intrinsic Euler--Lagrange--Poincar\'e--Herglotz equations]
\label{cor:ELP_Herglotz_clean}
Let $(x(t),a(t),z(t))$ satisfy $\dot x=\rho(a)$ and $\dot z=L(a,z)$ and
define $p=dL_{\mathrm{ver}}(a,z)$.  
Then $(x,a,z)$ satisfies the local Euler--Lagrange--Herglotz equations
if and only if $p$ satisfies the intrinsic equation
\begin{equation}\label{eq:ELP_Herglotz_connection}
\overline{\nabla}_{a}^{*}p (t) - \rho^{*}dL_{\mathrm{hor}}(a(t),z(t)) - \frac{\partial L}{\partial z}(a(t),z(t)) p(t) = 0,
\end{equation}
\end{corollary}

\begin{proof}
By Theorem~\ref{thm:ELP_Herglotz}, for any admissible curve
$(x(t),a(t),z(t))$ with $\dot x=\rho(a)$ and $\dot z=L(a,z)$, the local
Euler--Lagrange--Herglotz equations are equivalent to the intrinsic
relation
\[
\rho^{*}dL_{\mathrm{hor}}(a(t),z(t)) + \overline{\nabla}_{a}^{*}p (t) - \frac{\partial L}{\partial z}(a(t),z(t)) p(t) = 0,
\]
where $p(t)=dL_{\mathrm{ver}}(a(t),z(t))$.
This is precisely the statement of the corollary.
\end{proof}

\begin{remark}
The classical Hamilton--Pontryagin principle on a Lie algebroid
\cite{li2017lagrangian} is based on the variational stationarity of an
action functional.  In the Herglotz setting, no action functional exists:
the quantity $z(t)$ is defined dynamically by the contact-type evolution
equation $\dot z = L(v,z) + \langle p,a-v\rangle$, and variational
stationarity is imposed on the terminal value $z(T)$.
This single modification produces a contact deformation of all geometric
structures involved. 


Thus the Hamilton--Pontryagin--Herglotz principle provides a natural
contact-type extension of the classical Pontryagin variational principle
on Lie algebroids, retaining its geometric structure while incorporating
dissipation and nonequilibrium effects in a covariant manner.\hfill$\diamond$
\end{remark}

\begin{example}[Geodesic flow with Herglotz-type dissipation]
\label{ex:geodesic_Herglotz}
Let $(Q,g)$ be a Riemannian manifold and take $E=TQ$ with its standard Lie
algebroid structure ($\rho=\mathrm{id}$ and bracket the Lie bracket of
vector fields).  
We will denote by $\nabla^{g}$ the Levi--Civita connection of $g$ and by $\nabla$ a locally trivial connection used to derived the local expression of the intrinsic equations as in the proof of Lemma \ref{lem:local_ELP_Herglotz}. Consider the Herglotz Lagrangian
\[
L(q,\dot q,z)=\tfrac12\,g_q(\dot q,\dot q)-V(q)-\gamma z,
\qquad \gamma>0.
\]

Then, using the connection $\nabla$ we obtain
\[
p_q=dL_{\mathrm{ver}}(q,\dot q,z)=g_q(\dot q,\cdot),
\qquad
dL_{\mathrm{hor}}(q,\dot q,z)=-dV(q),
\qquad
\frac{\partial L}{\partial z}=-\gamma.
\]
And the intrinsic Euler--Lagrange--Poincar\'e--Herglotz equation becomes
\[
\dot{p}_q
= -\,dV(q) - \gamma\,p_q.
\]
Using the musical isomorphism $\sharp:T^{*}Q \to TQ$ defined by $g(\sharp(p_{q}), v_{q})=\langle p_{q}, v_{q} \rangle$ for any $v_{q}\in TQ$, we may conclude that  $\sharp(p_q)=\dot q$, and using that $\sharp(\dot{p}_{q})=\nabla^{g}_{\dot{q}}\dot{q}$ yields the
damped Newton--Lagrange equation
\[
\nabla^{g}_{\dot q}\dot q
= -\mathrm{grad}_g V(q)-\gamma\,\dot q,
\]
where $\mathrm{grad}_g V=\sharp(dV)$. Thus a Riemannian geodesic with potential acquires a natural linear damping term generated by the Herglotz variable, reproducing Rayleigh friction in a completely coordinate-free way.
\end{example}

\begin{example}[Euler--Poincaré--Herglotz equations]
Take the Lie algebroid $E=\mathfrak{g}$ with anchor $0$ and
bracket the Lie bracket on $\mathfrak g$.  
Since the base is a point, there is no horizontal part.  
If we start with the trivial $E$--connection $\nabla\equiv 0$, then $\overline{\nabla}$ is simply the adjoint representation:
\[
\overline{\nabla}_\xi X=[\xi,X]=\operatorname{ad}_\xi X,\qquad
\overline{\nabla}^{*}_\xi\mu=-\operatorname{ad}^*_\xi\mu.
\]

Let $\ell:\mathfrak g\times\mathbb R\to\mathbb R$ be a reduced Herglotz
Lagrangian, with reduced velocity $\xi(t)\in\mathfrak g$ and reduced
momentum $\mu=\partial\ell/\partial\xi(\xi,z)$.  
The Herglotz equation is $\dot z=\ell(\xi,z)$.

Applying Theorem~\ref{thm:ELP_Herglotz} gives the
\emph{Euler--Poincaré--Herglotz equation}
\[
\dot\mu
=
\operatorname{ad}^*_\xi\mu
+
\frac{\partial\ell}{\partial z}(\xi,z)\,\mu,
\]
a dissipative deformation of the classical Euler--Poincaré equation \cite{anahory2024reduction}.
\end{example}

\begin{example}[Charged particle in a magnetic field on an Atiyah algebroid]
\label{ex:magnetic_Herglotz_Atiyah}

Let $\pi:Q\to M$ be a principal $S^1$--bundle endowed with a principal
connection $\mathcal A:TQ \to \mathfrak{g}$ and curvature $B:TQ\times TQ \to \mathfrak{g}$,
representing a magnetic field on $M$.
The associated Atiyah algebroid is $E=TQ/S^1\to M$. Using the connection $\mathcal A$, one has the standard Lie algebroid identification
\[
E \cong TM \oplus (Q\times\mathbb R)/S^1,
\]
where $\tilde{\mathfrak{g}}:=(Q\times\mathbb R)/S^1$ is the adjoint bundle over $M$. A typical element $a\in E_x$ is written as $a=(X,\bar{u})$, where $X\in T_xM$ represents the horizontal velocity and $\bar{u}=[(q,u)]$ with $u\in\mathbb R$ and $\pi(q)=x$ is the internal charge variable. The projection is simply $\tau(X,\bar{u})=\tau_M(X)$, where $\tau_M:TM\to M$ is the tangent bundle projection, the anchor is the projection onto the first component $\rho(X,\bar{u})=X$ and the bracket is given by
$$[(X,\bar{u}), (Y,\bar{v})]=\left( [X,Y], \tilde{\nabla}_X \bar{v} - \tilde{\nabla}_Y \bar{u} - \tilde{B}(X,Y)\right),$$
where $\tilde{\nabla}$ is the associated connection on $\tilde{\mathfrak{g}}$ to the $E$-connection $\nabla$ derived from the principal connection $\mathcal{A}$ and $\tilde{B}:TM\times TM\to \mathfrak{g}$ is the reduced principal connection curvature defined by $\tilde{B}(X,Y)=[(q, B(X^{h}, Y^{h}))]$, for $X,Y\in T_xM$ and $X^{h}$ denotes the horizontal lift (see \cite{cendra2001lagrangian} for more details).



Let $g$ be a Riemannian metric on $M$, $m>0$ and $e\in\mathbb R$ be the mass and electric charge of the particle and $V:M\to\mathbb{R}$ a potential function defined on $M$.
Consider the Herglotz Lagrangian
\[
L(x,\dot{x},u,z)
=\tfrac m2\,g_x(\dot{x},\dot{x})+e\,u -V(x)-\gamma z,
\qquad \gamma>0.
\]

With respect to a $TM$--connection $\nabla$ on $E$ of the form
$$\nabla_{X}(Y,\bar{u})=(\nabla^{M}_X Y, \tilde{\nabla}_{X} \bar{u}),$$
where $\nabla^{M}_X Y$ is a locally trivial affine connection on $M$ and $\tilde{\nabla}$ is the associated connection on $\tilde{\mathfrak{g}}$ defined previously, the vertical differential of $L$ is
\[
dL_{\mathrm{ver}}(x,\dot{x},u,z)
=(m\,g_x(\dot{x},\cdot),\,e)\in T_x^*M\oplus\mathbb R^*,
\]
which defines the momentum
\[
p=(p_{x},p_u)=(m\,g(\dot{x},\cdot),\,e).
\]
The horizontal differential and partial derivative with respect to $z$ are
\[
dL_{\mathrm{hor}}=-dV,
\qquad
\frac{\partial L}{\partial z}=-\gamma.
\]
In addition, the connection $\overline{\nabla}$ on $E$ satisfies
$$\overline{\nabla}_{(X,\bar{u})}(Y,\bar{v}) = \nabla_{(X,\bar{u})}(Y,\bar{v}) - (0, \tilde{B}(X,Y)).$$
Applying the intrinsic Euler--Lagrange--Poincar\'e--Herglotz equation
\eqref{eq:ELP_Herglotz_connection} yields
\[
\overline{\nabla}_{(\dot{x},\bar{u})}^{*}p
=
-\rho^*(dV)
-
\gamma p.
\]
Given a time-dependent section of $E^{*}$ such as $p(t)=\mu(t, x(t))$, we have that
$$\langle \overline{\nabla}_{(\dot{x},\bar{u})}^{*}p (t), (Y,\bar{v}) \rangle = \langle\dot{p}, (Y,\bar{v}) \rangle - \langle p(t), \overline{\nabla}_{(\dot{x},\bar{u})} (Y,\bar{v}) \rangle = \langle\dot{p}, (Y,\bar{v}) \rangle  - \langle p(t), \nabla_{Y} (\dot{x},\bar{u}) + [(\dot{x},\bar{u}), (Y,\bar{v})]  \rangle$$
which gives
$$\langle \overline{\nabla}_{(\dot{x},\bar{u})}^{*}p (t), (Y,\bar{v}) \rangle = \langle\dot{p}, (Y,\bar{v}) \rangle  - \langle p, (\overline{\nabla}^{M}_{X} Y, \tilde{\nabla}_X\bar{v} - \tilde{B}(X,Y) )\rangle.$$
Projecting onto the component along $TM$ and using that the connection $\overline{\nabla}^{M}$ is locally trivial, we obtain a coordinate expression of the form
\[
\dot{p}_{x} + \tilde{B}(\dot{x},\cdot) p_u= -dV-\gamma p_x.
\]
where $(0 ,\tilde{B}(\dot{x},\cdot) )^{*}$ is an adjoint operator defined by $\langle(0 ,\tilde{B}(\dot{x},\cdot) )^{*}p, Y\rangle = \langle p, (0, \tilde{B}(\dot{x},Y)) \rangle$.
Using the musical isomorphism $\sharp:T^{*}M \to TM$, we obtain the \emph{magnetic geodesic equation with Herglotz damping}
\[
m\,\nabla^g_{\dot{x}} \dot{x}
=
-e\,\sharp(\iota_{\dot{x}}\tilde{B})
-
\mathrm{grad}_g V(x)
-
\gamma m\,\dot{x},
\qquad
\dot z=L(x,\dot{x},u,z).
\]
where $\mathrm{grad}_g V=\sharp(dV)$, $(\iota_{\dot{x}}\tilde{B})(Y)=\tilde{B}(\dot{x},Y)$ and $\nabla^g$ is the Levi--Civita connection of $(M,g)$.

Thus the curvature of the Atiyah algebroid reproduces the magnetic Lorentz
force, while the Herglotz term produces a linear dissipation of kinetic
energy in a fully intrinsic and gauge-invariant manner.
\end{example}

\section{Conclusions and Future Work}

In this paper we developed a variational framework for \emph{dissipative mechanics on Lie algebroids}. We derived the Euler--Lagrange--Herglotz equations in local coordinates and then obtained a fully coordinate-free, connection-based formulation of the Euler--Lagrange--Poincar\'e--Herglotz equations. This unifies classical contact mechanics, Euler--Poincar\'e and Lagrange--Poincar\'e reduction with linear dissipation.

We also introduced the Hamilton--Pontryagin--Herglotz principle on a Lie algebroid, yielding an implicit system that captures simultaneously kinematic constraints, the momentum relation, the horizontal dynamics, and the Herglotz evolution equation. This formulation is well adapted to generalizations to discrete geometry and to the systematic treatment of symmetries, reduction, and curvature effects. Moreover, our framework leads naturally to energy balance laws and Noether--Herglotz theorems, showing how classical conserved quantities are replaced by \emph{dissipated invariants} determined by an integrating factor.

The geometric nature of our results suggests several important research directions. A central objective is the construction of \emph{geometric structure-preserving contact integrators} on \emph{Lie groupoids}, obtained by discretizing the Hamilton--Pontryagin--Herglotz principle. Such integrators would preserve admissibility at the groupoid level, reproduce the correct exponential energy decay, and retain the dissipated Noether quantities exactly. This represents the natural dissipative analogue of variational integrators on groupoids and remains largely unexplored.

\appendix

\section{Appendix: Local expression of the intrinsic Herglotz equation}

\begin{lemma}\label{lem:local_ELP_Herglotz}
Let $\{e_\alpha\}$ be a local frame of $E$ with structure functions
$[e_\alpha,e_\beta]=C^\gamma_{\ \alpha\beta}e_\gamma$ and anchor
components $\rho(e_\alpha)=\rho^i_{\ \alpha}\frac{\partial}{\partial x^{i}}$.
Write $a(t)=y^\alpha(t)e_\alpha|_{x(t)}$ and
$p(t)=p_\alpha(t)e^\alpha|_{x(t)}$ where $e^\alpha$ is the dual frame.
Then the intrinsic Euler--Lagrange--Poincar\'e--Herglotz equation
\[
\overline{\nabla}_{a}^{*}p (t) - \rho^{*}dL_{\mathrm{hor}}(a(t),z(t))  - \frac{\partial L}{\partial z}(a(t),z(t)) p(t) = 0, \quad p=dL_{\mathrm{ver}}
\]
takes the local form
\[
\dot p_\alpha
+ C^\gamma_{\ \alpha\beta}(x)\,y^\beta p_\gamma
-
\rho^i_{\ \alpha}(x)\frac{\partial L}{\partial x^i}
-\frac{\partial L}{\partial z}(x,y,z)\,p_\alpha = 0,
\]
which is exactly the Euler--Lagrange--Herglotz system of
Theorem~\ref{thm:ELH_algebroid}.
\end{lemma}

\begin{proof}

Choosing without loss of generality the locally trivial $TM$-connection on $E$, i.e., the connection with $\nabla_{\frac{\partial}{\partial x^{i}}} e_{\alpha}\equiv 0$, the vertical and horizontal components of $dL$ are
$$dL_{\mathrm{ver}}=\frac{\partial L}{\partial y^{\alpha}}e^{\alpha}, \quad dL_{\mathrm{hor}}=\frac{\partial L}{\partial x^{i}}dq^{i}.$$

In addition, given a time-dependent section of $E^{*}$ such as $p(t)=\mu(t, x(t))$, we have that
$$\langle \overline{\nabla}_{a}^{*}p (t), e_{\alpha} \rangle = \dot{p}_{\alpha} - \langle p(t), \overline{\nabla}_{a} e_{\alpha} \rangle = \dot{p}_{\alpha} - \langle p(t), \nabla_{\rho(e_{\alpha})} a + [a, e_{\alpha}]  \rangle = \dot{p}_{\alpha} - y^{\beta} C_{\beta \alpha}^{\gamma} p_{\gamma} = \dot{p}_{\alpha} + y^{\beta} C_{\alpha \beta}^{\gamma} p_{\gamma}.$$

Substituting these expressions into the intrinsic Euler-Lagrange-Poincaré-Herglotz equations we obtain the Euler--Lagrange--Herglotz system of
Theorem~\ref{thm:ELH_algebroid}, i.e.,
$$\frac{d}{dt}\left(\frac{\partial L}{\partial y^\alpha}\right)
+ C^\gamma_{\ \alpha\beta}(x)\,y^\beta \frac{\partial L}{\partial y^\gamma}
- \rho^i_{\ \alpha}(x)\,\frac{\partial L}{\partial x^i}
- \frac{\partial L}{\partial z}\,\frac{\partial L}{\partial y^\alpha}
= 0.$$
It is easy to repeat the calculations and check that if the connection was not locally trivial and we had instead $\nabla_{\frac{\partial}{\partial x^{i}}} e_{\alpha}= \Gamma_{i\alpha}^{\gamma} e_{\gamma}$, the terms containing Christoffel symbols would cancel out and we would obtain the same local expressions.
\end{proof}

\bibliography{mybib}
\bibliographystyle{ieeetr}

\end{document}